\documentclass{article}
\usepackage[a4paper]{geometry}
\usepackage[myheadings]{fullpage}
\usepackage{fancyhdr}
\usepackage{graphicx, wrapfig, subcaption, setspace, booktabs}
\usepackage{fourier}
\usepackage[english]{babel}
\usepackage{sectsty}
\usepackage{amsthm}
\usepackage{enumerate}
\usepackage{bigints}
\usepackage{bm}
\usepackage{appendix}
\usepackage{leftidx}
\usepackage[numbers]{natbib}

\newcommand{\norm}[1]{\int_{0}^{\infty}{#1}}
\newcommand{\normOne}[1]{\left\lVert#1\right\rVert_{1}}
\newcommand{\normTwo}[1]{\left\lVert#1\right\rVert_{2}}
\newcommand{\normInf}[1]{\left\lVert#1\right\rVert_{\infty}}
\newcommand{\integralInf}[1]{\int_{0}^{\infty}{#1}}
\newcommand{\normOp}[1]{\left\lVert#1\right\rVert_{\textnormal{op}}}

\newcommand{\trans}[1]{\leftidx{^\dagger}{#1}}
\newcommand{\spectralRadius}[1]{\rho{(#1)}}

\newcommand{\modulus}[1]{\mid{#1}\mid}
\newcommand{\matfunction}[2]{\mathcal{F}(\mathcal{M}_{#1}(#2))}
\newtheorem{theorem}{Theorem}
\newtheorem{proposition}{Proposition}
\newtheorem{definition}{Definition}
\newtheorem{assumption}{Assumption}

\newtheorem{remark}{Remark}
\newtheorem{corollary}{Corollary}
\newtheorem{lemma}{Lemma}
\title{From microscopic price dynamics to \\ multidimensional rough volatility models \thanks{Mathieu Rosenbaum and Mehdi Tomas gratefully acknowledge financial support of the ERC 679836 Staqamof and the Chair Analytics and Models for Regulation. Mehdi Tomas gratefully acknowledges the support of the Chair Econophysics and Complex Systems. The authors thank Michael Benzaquen and Iacopo Mastromatteo for their helpful comments and are grateful to Eduardo Abi Jaber, Jean-Philippe Bouchaud, Antoine Fosset and Paul Jusselin for very fruitful discussions and suggestions.}}
\author{Mathieu Rosenbaum\footnote{CMAP, \'Ecole Polytechnique, mathieu.rosenbaum@polytechnique.edu} \and Mehdi Tomas \footnote{CMAP \& LadHyx, \'Ecole Polytechnique, mehdi.tomas@polytechnique.edu}}
\date{\today}

\begin{document}

\maketitle

\begin{abstract}
Rough volatility is a well-established statistical stylised fact of financial assets. This property has lead to the design and analysis
of various new rough stochastic volatility models. However, most of these developments have been carried out in the mono-asset case.
In this work, we show that some specific multivariate rough volatility models arise naturally from microstructural properties of the joint dynamics
of asset prices. To do so, we use Hawkes processes to build microscopic models that reproduce accurately high frequency
cross-asset interactions and investigate their long term scaling limits. We emphasize
the relevance of our approach by providing insights on the role of microscopic features such as momentum and mean-reversion
on the multidimensional price formation process. We in particular recover classical properties of high-dimensional stock correlation matrices.
\end{abstract}



\section{Introduction}
\label{sec:introduction}

\subsection{A microstructural viewpoint on rough volatility}

It is now widely accepted that volatility is rough (see \citep{GJR14} and among others \cite{DaFonseca2019VolatilityRough,Livieri2018RoughPrices}): the log-volatility process is well-approximated by a fractional Brownian motion with small Hurst parameter $H \approx 0.1$, which corresponds to H\"older regularity of order $H-\epsilon$, $\epsilon>0$. Furthermore, rough volatility models capture key features of the implied volatility surface and its dynamics (see \cite{Bayer2015PricingVolatility, ElEuch2018RougheningHeston, Horvath2019DeepVolatility}).
\\ \\
The macroscopic phenomenon of rough volatility is seemingly universal: it is observed for a large class of financial assets and across time periods. This universality may stem from fundamental properties such as market microstructure or no arbitrage. This raised interest in building \textit{microscopic} models for market dynamics which reproduce rough volatility at a \textit{macroscopic} scale. For us, the microscopic time scale is of the order of milliseconds, where asset prices are jump processes, while the macroscopic scale is approximately of the order of days, where asset prices appear essentially continuous.
\\ \\
Hawkes processes, first introduced in \cite{Hawkes1971PointProcesses., Hawkes1971SpectraProcesses, Hawkes1974AProcess} to model earthquake aftershocks, are nowadays very popular to model the high frequency dynamics of prices of financial assets (see \cite{Bacry2015HawkesFinance} for an overview of applications). In particular, the papers \cite{ElEuch2018TheVolatility., Jaisson2015LimitProcesses, Jaisson2016RoughProcesses} successfully establish a link between rough volatility and history dependent Hawkes-type point processes which reproduce:
\begin{enumerate}[(i)]
    \itemsep0em
    \item the no statistical arbitrage property: it is very hard to design strategies which are on average profitable at the high frequency scale;
    \item the long memory property of order flow due to the splitting of large orders (meta-orders) into smaller orders;
    \item the high degree of endogeneity of financial markets: the large majority of market activity (including price moves, cancellations and market and limit orders) occurs in response to previous market activity (as opposed to exogenous information such as news).
\end{enumerate}
We refer to \cite{ElEuch2018TheVolatility., Hardiman2013CriticalAnalysis} for details about these three stylised facts. This Hawkes-based microscopic framework can easily account for other features of markets: for example \cite{Jusselin2018No-ArbitrageVolatility} examines the issue of permanent market impact, \cite{ElEuch2018TheHeston} studies how a bid/ask asymmetry creates a negative price/volatility correlation, while the so-called Zumbach effect is considered in \cite{Dandapani2019FromEffect}.
\\ \\
Inspired by \cite{ElEuch2018TheVolatility.,Jaisson2015LimitProcesses, Jaisson2016RoughProcesses}, the goal of this paper is to use Hawkes processes to find a micro-founded setting of multivariate rough volatility which:
\begin{enumerate}[(i)]
    \itemsep0em
    \item enforces no statistical arbitrage between multiple assets;
    \item is consistent with the long memory property of the order flow and the high degree of endogeneity of financial markets;
    \item explains stylised facts from the microscopic price formation process, with a focus on the structure of high-dimensional stock correlation matrices.
\end{enumerate}
This approach enables us to characterise the type of price dynamics arising from those constraints. Readers interested in multivariate rough volatility may consult \cite{Cuchiero2019MarkovianProcesses} for general construction of a class of affine multivariate rough covariance models. Our goal is more modest here: we are interested in finding macroscopic dynamics originating from microscopic insights, not in a full mathematical analysis of the class of possible models for multivariate rough volatility.
Note also that in the concomitant work \cite{Jaber2019AType}, the authors study weak solutions of stochastic Volterra equations in a very comprehensive framework. Some of our technical results can be derived from their general approach. In our setting, we rather provide simple and natural proofs inspired from  \cite{ElEuch2018TheVolatility.,Jaisson2015LimitProcesses, Jaisson2016RoughProcesses}, allowing us to emphasize financial interpretations of the results, which is the core of this work.

\subsection{Modeling endogeneity of financial markets}
We first introduce the asymptotic framework which models the high endogeneity of financial markets in the mono-asset case (as in \cite{Bacry2013ModellingProcesses, ElEuch2018TheVolatility., Jaisson2015LimitProcesses, Jaisson2016RoughProcesses}) for clarity purposes before moving to the multivariate setting of interest. At the high frequency scale, the price is a piecewise constant process with upward and downward jumps captured by a bi-dimensional counting process $\bm{N} = (N^{1+}, N^{1-})$, with $N^{1+}$ counting the number of upward price moves and $N^{1-}$ the number of downward price moves. Assuming that all jumps are of the same size, the microscopic price of the asset is the difference of the number of upward and downward jumps (where the initial price is set to zero for simplicity) and therefore can be written
$$
P_t = N^{1+}_t - N^{1-}_t.
$$
Our assumption is that $\bm{N}$ is a Hawkes process with intensity $\bm{\lambda} = (\lambda^{1+}, \lambda^{1-})$ such that
\begin{align*}
\lambda^{1+}_t &= \mu^{1+}_t + \int_{0}^{t} \phi_{1+,1+}(t-s) dN^{1+}_s + \int_{0}^{t} \phi_{1+,1-}(t-s) dN^{1-}_s \\
\lambda^{1-}_t &= \mu^{1-}_t + \int_{0}^{t} \phi_{1-,1+}(t-s) dN^{1+}_s + \int_{0}^{t} \phi_{1-,1-}(t-s) dN^{1-}_s,
\end{align*}
where the $\bm{\mu} \colon \mathbb{R}_+ \to \in \mathbb{R}^2_{+}$ is called the \textit{baseline} and $\bm{\phi} \colon \mathbb{R}_+ \to \mathcal{M}_{2}(\mathbb{R}_{+})$ is called the \textit{kernel}, where we write vectors and matrices in bold and $\mathcal{M}_{n,m}(X)$ (resp. $\mathcal{M}_{n}(X)$) for the set of $X$-valued $n \times m$ (resp. $n \times n$) matrices. From a financial perspective, we can easily interpret the different terms above:
\begin{itemize}
    \item on the one hand, $\mu_1^+$ (resp. $\mu_1^-$) is an exogenous source of upward (resp. downward) price moves;
    \item on the other hand, $\bm{\phi}$ is an endogenous source of price moves. For example, $\phi_{1+,1-}$ increases the intensity of upward price jumps after a downward price jump, creating a mean-reversion effect (while $\phi_{1+,1+}$ creates a trending effect).
\end{itemize}
To further encode the long-memory property of the order flow, \cite{ElEuch2018TheVolatility., Jaisson2015LimitProcesses} consider heavy-tailed kernels where, writing $\spectralRadius{\bm{M}}$ for the spectral radius of a matrix $\bm{M}$, for some $c>0$ and $\alpha \in (1/2,1)$ we have
$$
\rho \big (\int_t^{\infty} \bm{\phi}(s)ds \big) \underset{t \to \infty}{\sim} c t^{-\alpha}.
$$
Such a model satisfies the stability property of Hawkes processes (see for example \cite{Jaisson2015LimitProcesses}) as long as $\spectralRadius{\normOne{\bm{\phi}}}<1$ (writing $\normOne{\cdot}$ for the $L^1$ norm). In fact, calibration of Hawkes processes on financial data suggests that this stability condition is almost violated. To account for this effect, the authors of \cite{ElEuch2018TheVolatility., Jaisson2015LimitProcesses} model the market up to time $T$ with a Hawkes process $\bm{N^{T}}$ of baseline $\bm{\mu^{T}}$ and kernel $\bm{\phi^{T}}$. The microscopic price until time $T$ is then
$$
P^{T,1}_t = N^{T, 1+}_t - N^{T, 1-}_t.
$$
In order to obtain macroscopic dynamics, the time horizon must be large, thus the sequence $T_n$ tends towards infinity (from now on, we write $T$ for $T_n$). As $T$ tends to infinity, $\bm{\phi^{T}}$ almost saturates the stability condition: $\underset{n \to \infty}{\lim}\spectralRadius{\normOne{\bm{\phi}^{T}}} = 1$. A macroscopic limit then requires scaling the processes appropriately to obtain a non-trivial limit.  Details on the proper rescaling of the processes are given in Section \ref{sec:main_results_org}.

\subsection{Multivariate setting}
Having described the asymptotic setting in the mono-asset case, we now model $m$ different assets. The associated counting process is now a $2m$-dimensional process $\bm{N^T} = (N^{T,1+}, N^{T,1-}, N^{T,2+}, \dots, N^{T,m-})$ and its intensity satisfies
$$
\bm{\lambda_t^T} = \bm{\mu^T} + \int_{0}^{t} \bm{\phi(t-s)^T} d\bm{N_s^T}.
$$
The counting process $\bm{N}$ includes the upward and downward price jumps of $m$ different assets and the microscopic price of Asset $i$, where $1 \leq i \leq m$, is simply
$$
P^{T,i}_t = N^{T,i+}_t - N^{T,i-}_t.
$$
This allows us to capture correlations between assets since, focusing for example on Asset $1$, we have
$$
\lambda_{t}^{T,1+} = \mu^{T,1+}_t + \int_{0}^{t} \phi_{1+,2+}^T(t-s) dN^{T,2+}_s + \int_{0}^{t} \phi_{1-,2+}^T(t-s) dN^{T,2+}_s + \cdots.
$$
Therefore $\phi^T_{1+,2+}$ increases the intensity of upward jumps on Asset 1 after an upward jump of Asset 2 while $\phi^T_{1-,2+}$ increases the intensity of downward jumps, etc.
\\ \\
We now need to adapt the nearly-unstable setting to the multidimensional case. Thus we have to find how to saturate the stability condition and to translate the long memory property of the order flow.
\\ \\
In \cite{ElEuch2018TheVolatility.}, $\bm{\phi^T(t)}$ is taken diagonalisable (in a basis independent of $T$ and $t$) with a maximum eigenvalue $\xi^T(t)$ such that $\underset{T \to \infty}{\lim} \normOne{\xi^T} = 1$. However this structure leads to the same volatility for all assets and thus cannot be a satisfying solution for realistic market dynamics. We take here a sequence of trigonalisable (in a basis $\bm{O}$ independent of $T$ and $t$) kernels $\bm{\phi^T(t)}$ with $n_c >0$ eigenvalues almost saturating the stability condition. Thus the Hawkes kernel is taken of the form (using block matrix notation in force throughout the paper)
$$
\bm{\phi^T(t)} = \bm{O}\begin{pmatrix}
    \bm{A^T(t)} & \bm{0} \\
    \bm{B^T(t)} & \bm{C^T(t)}
    \end{pmatrix} \bm{O}^{-1},
$$
where $\bm{A^T}\colon \mathbb{R_+} \to \mathcal{M}_{n_c}(\mathbb{R})$, $\bm{B^T} \colon \mathbb{R_+} \to \mathcal{M}_{2m-n_c,n_c}(\mathbb{R})$ and $\bm{C^T} \colon \mathbb{R_+} \to \mathcal{M}_{2m-n_c}(\mathbb{R})$. Note that we will see that in the limit, macroscopic volatilities and prices are independent of the chosen basis. We assume that the stability condition is saturated at the speed $T^{-\alpha}$ where $\alpha \in (1/2,1)$ is again related to the tail of the matrix kernel (see below). The saturation condition translates to
$$
T^{\alpha} \big( \bm{I} - \norm{\bm{A^T}} \big) {\underset{T \to \infty}{\to}} \bm{K},
$$
where $\bm{K}$ is an invertible matrix.
\\ \\
We now need to encode the long memory property of the order flow. We can expect orders to be sent jointly on different assets (this can be due, for example, to portfolio rebalancing, risk management or optimal trading) and split under different time scales depending on idiosyncratic components (such as daily traded volume or volatility). Empirically the approximation that despite idiosyncrasies a common time scale for order splitting exists is partially justified: for example \cite{Benzaquen2017DissectingAnalysis} shows that market impact, which is directly related to the order flow, is well-approximated by a single time scale for many stocks. Finally, this property is encoded by imposing a heavy-tail condition for $\bm{A} := \underset{T \to \infty}{\lim} \bm{A^T}$ with the previous exponent $\alpha$:
$$
\alpha x^{\alpha} \int_x^\infty  \bm{A(s)}ds \underset{x \to \infty}{\to} \bm{M},
$$
with $\bm{M}$ an invertible matrix.
\subsection{Main results and organization of the paper}
\label{sec:main_results_org}

In the framework described above, we show that the macroscopic limit of prices is a multivariate version of the rough Heston model introduced in \cite{ElEuch2018TheHeston, ElEuch2018RougheningHeston}, where the volatility process is a solution of a multivariate rough stochastic Volterra equation. Thus we derive a natural multivariate setting for rough volatility using nearly-unstable Hawkes processes.
\\ \\
More precisely, define the rescaled processes (see \cite{Jaisson2015LimitProcesses} for details), for $t \in [0,1]$:
\begin{align}
    \label{eq:definition_rescaled_processes_X}
    \bm{X^T_t} &:= \dfrac{1}{T^{2\alpha}} \bm{N^T_{tT}} \\
    \label{eq:definition_rescaled_processes_Y}
    \bm{Y^T_t} &:= \dfrac{1}{T^{2\alpha}} \int_0^{tT} \bm{\lambda_s} ds  \\
    \label{eq:definition_rescaled_processes_Z}
    \bm{Z^T_t} &:= T^{\alpha} (\bm{X^T_t} - \bm{Y^T_t}) = \dfrac{1}{T^{\alpha}} \bm{M^T_{tT}} \\
    \label{eq:definition_rescaled_processes_P}
    \bm{P^{T}_{t}} &= \dfrac{1}{T^{2\alpha}} (N^{T,1+}_{tT} - N^{T,1-}_{tT}, \cdots, N^{T,m+}_{tT} - N^{T,m-}_{tT}).
\end{align}
We refer to $\bm{P^{T}}$ as the (rescaled) microscopic price process. Under some additional technical and no statistical arbitrage assumptions, there exists an $n_{c}$ dimensional process $\bm{\Tilde{V}}$, matrices $\bm{\Theta^1} \in \mathcal{M}_{n_{c}}(\mathbb{R}), \bm{\Theta^2} \in \mathcal{M}_{n-n_{c}}(\mathbb{R}), \bm{\Lambda_{0}} \in \mathcal{M}_{n_{c}}(\mathbb{R}), \bm{\Lambda_{1}} \in \mathcal{M}_{n_{c}}(\mathbb{R}), \bm{\Lambda_{2}} \in \mathcal{M}_{n_{c},n-n_{c}}(\mathbb{R}), \bm{\theta_{0}} \in \mathbb{R}^{n_{c}}$ and a Brownian motion $\bm{B}$ such that
\begin{itemize}
    \itemsep0em
    \item Any macroscopic limit point $\bm{P}$ of the sequence $\bm{P^T}$ satisfies
    $$
    \bm{P_t} = (\bm{I} + \bm{\Delta}) \trans{\bm{Q}} \int^t_0 \textnormal{diag}(\sqrt{\bm{V_s}}) d \bm{B_s},
    $$
    where $\bm{Q} := (\bm{e_{1}} - \bm{e_{2}} \mid \cdots \mid  \bm{e_{2m-1}} - \bm{e_{2m}})$, writing $\trans{\bm{Q}}$ for the transpose of $\bm{Q}$, $(\bm{e_i})_{1 \leq i \leq 2m}$ for the canonical basis of $\mathbb{R}^{2m}$ and $\bm{\Delta} = (\Delta_{ij})_{1 \leq i,j \leq m} \in \mathcal{M}_{m}(\mathbb{R})$ is defined in Section \ref{sec:main_results} while $\bm{V}$ is defined below.
    \item $\bm{\Theta^1}\bm{\Tilde{V}} = (V^{1}, \cdots, V^{n_c})$ and $\bm{\Theta^2}\bm{\Tilde{V}} = (V^{n_c+1}, \cdots, V^{n})$.
    \item $\bm{\Tilde{V}}$ has H\"older regularity $\alpha - 1/2 - \epsilon$ for any $\epsilon > 0$.
    \item For any $t$ in $[0,1]$, $\bm{\Tilde{V}}$ satisfies
    \begin{equation*}
    \bm{\Tilde{V}_t} =  \int_0^t (t-s)^{\alpha-1}(\bm{\theta_0} - \bm{\Lambda_{0}} \bm{\Tilde{V}_s}) ds  +\int_0^t (t-s)^{\alpha-1} \bm{\Lambda_{1}}\textnormal{diag}(\sqrt{\bm{\Theta^1} \bm{\Tilde{V}_s}}) d\bm{W_s} + \int_0^t (t-s)^{\alpha-1} \bm{\Lambda_{2}}\textnormal{diag}(\sqrt{\bm{\Theta^2} \bm{\Tilde{V}_s}}) d\bm{Z_s},
\end{equation*}
 where $\bm{W} := (B^1, \cdots, B^{n_c})$, $\bm{Z} := (B^{n_c+1}, \cdots, B^{n})$ and we write $\sqrt{\bm{x}}$ for the component-wise square root of vectors of non-negative entries. 
\end{itemize}
Thus the volatility process $\bm{V}$ is driven by $\bm{\Tilde{V}}$, which represents volatility factors, of which there are as many as there are critical directions.
\\ \\
We can use this result to provide microstructural foundations for some empirical properties of correlation matrices. Informally, considering that our assets have similar self-exciting features in their microscopic dynamics, we show that any macroscopic limit point $\bm{P}$ of the sequence $\bm{P^T}$ satisfies $\bm{P}$
\begin{align*}
    \bm{P_t} &= \bm{\Sigma} \int_0^t \textnormal{diag}(\sqrt{\bm{V_s}}) d\bm{W_s},
\end{align*}
where $\bm{W}$ is a Brownian motion, $\bm{V}$ satisfies a stochastic Volterra equation and $\bm{\Sigma}$ has one very large eigenvalue followed by smaller eigenvalues that we can interpret as due to the presence of sectors and a bulk of eigenvalues much smaller than the others. This is typical of actual stock correlation matrices (see for example \cite{Laloux1999NoiseMatrices} for an empirical study).
\\ \\
The paper is organised as follows. Section \ref{sec:assumptions} rigorously introduces the technical framework sketched in the introduction. We present and discuss the main results in Section \ref{sec:main_results} which are then applied in examples developed in Section \ref{sec:example}. Proofs can be found in Section \ref{sec:proofs} while some technical results, including proofs of the various applications, are available in an appendix.

\section{Assumptions}
\label{sec:assumptions}

Before presenting the main results, we make precise the framework sketched out in the introduction. Different examples of Hawkes processes satisfying our assumptions are given in Section \ref{sec:example}. 
\\ \\
Consider a sequence of measurable functions $\bm{\phi^T} \colon \mathbb{R}_{+} \to \mathcal{M}_{2m}(\mathbb{R}_{+})$ and $\bm{\mu^T} \colon \mathbb{R}_{+} \to \mathbb{R}^{2m}_{+}$, where the pair $(\bm{\mu^T}, \bm{\phi^T})$ will be used to model the market dynamics until time $T$ via a Hawkes process $\bm{N^T}$ of baseline $\bm{\mu^T}$ and kernel $\bm{\phi^T}$. Each kernel $\bm{\phi^T}$ is stable ($\rho \big( \normOne{\bm{\phi^T}} \big)<1$).

\begin{assumption}
\label{ass:StructureHawkes}
There exists $\bm{O}$ an invertible matrix such that each $\bm{\phi^T}$ can be written as
\begin{equation*}
    \label{eq:shapePhiT}
    \bm{\phi^T} = \bm{O} \begin{pmatrix}
    \bm{A^T} & \bm{0} \\
    \bm{B^T} & \bm{C^T}
    \end{pmatrix} \bm{O}^{-1},
\end{equation*}
where $\bm{A^T}\colon \mathbb{R_+} \to \mathcal{M}_{n_c}(\mathbb{R})$, $\bm{B^T} \colon \mathbb{R_+} \to \mathcal{M}_{2m-n_c,n_c}(\mathbb{R})$, $\bm{C^T} \colon \mathbb{R_+} \to \mathcal{M}_{2m-n_c}(\mathbb{R})$. Furthermore, the sequence $\bm{\phi^T}$ converges towards $\bm{\phi} \colon \mathbb{R}_{+} \to \mathcal{M}_{2m}(\mathbb{R}_{+})$ as T tends to infinity and, writing $\bm{A}, \bm{B}, \bm{C}$ for the limits of $\bm{A^T}, \bm{B^T}, \bm{C^T}$ as T tends to infinity, $\spectralRadius{\norm{\bm{C}}} < 1$. 
\\ \\ 
Additionally, there exists $\alpha \in (1/2,1)$, $\bm{K},\bm{M}$ invertible matrices and $\bm{\mu} \colon [0,1] \to \mathbb{R}_+$ such that
\begin{align}
    T^{\alpha} (\bm{I} - \normOne{\bm{A^T}}) & {\underset{T \to \infty}{\to}} \bm{K}     \label{eq:convergence_rates_AT}
     \\
    \alpha x^{\alpha} \int_x^\infty  \bm{A(s)}ds & \underset{x \to \infty}{\to} \bm{M}     \label{eq:convergence_rates_A}
 \\
    T^{1-\alpha} \bm{\mu^T_{tT}} \underset{T \to \infty}{\to} \bm{\mu_{t}},     \label{eq:convergence_rates_mu}
\end{align}
where $\bm{K}\bm{M}^{-1}$ has strictly positive eigenvalues.

\end{assumption}

Realistic market dynamics require enforcing no statistical arbitrage conditions on the kernels, as in the spirit of \cite{Jaisson2015LimitProcesses}. To determine which conditions need to be satisfied to prevent such arbitrage, we write the intensity of the counting process $\bm{\lambda^T}$ using the compensator process $\bm{M^T_t} := \bm{N^T_t} - \bm{\lambda^T_t}$ and  $\bm{\psi^T} = \sum_{k \geq 1} (\bm{\phi^{T}})^{*k}$ (see for example Proposition 2.1 in \cite{Jaisson2015LimitProcesses}). We have
\begin{equation}
    \label{eq:role_delta}
        \bm{\lambda^T_t} = \bm{\mu^T} + \int_0^t \bm{\psi^T(t-s)} \bm{\mu^T_s} ds + \int_0^t \bm{\psi^T(t-s)} d\bm{M^T_s}.
\end{equation}
Thus, the expected intensities of upward and downward price jumps of Asset $i$ are
\begin{align*}
    \mathbb{E}[\lambda^{T,i+}_t] &= \mu^{T,i+}_{t} + \sum_{1 \leq j \leq 2m} \int_0^t \psi^T_{i+,j-}(t-s) \mu^{T,j-}_s ds + \sum_{1 \leq j \leq 2m} \int_0^t \psi^T_{i+,j+}(t-s) \mu^{T,j+}_s ds \\
    \mathbb{E}[\lambda^{T,i-}_t] &= \mu^{T,i-}_{t} + \sum_{1 \leq j \leq 2m} \int_0^t \psi^{T}_{i-,j-}(t-s) \mu^{T,j-}_s ds + \sum_{1 \leq j \leq 2m} \int_0^t \psi^{T}_{i-,j+}(t-s) \mu^{T,j+}_s ds.
\end{align*}
The above leads us to the following assumption.
\begin{assumption}
\label{ass:no_arb}
For any $1 \leq i,j \leq m$:
\begin{enumerate}[(i)]
    \item $\psi^{T}_{i+,j+} + \psi^{T}_{i+,j-} = \psi^{T}_{i-,j+} + \psi^{T}_{i-,j-}$ (no pair-trading arbitrage) \label{ass:no_pair_trading_arbitrage}
    \item $\underset{T \to \infty}{\lim}\Big( \norm{\psi^{T}_{i+,j}} - \norm{\psi^{T}_{i+,j+}} \Big) < \infty$ (suitable asymptotic behaviour of the intensities)
\end{enumerate}
\end{assumption}

Under the above conditions and if $\mu^{T,i+} = \mu^{T,i-}$ for all $1 \leq i \leq m$, then $\mathbb{E}[\lambda^{T,i+}_t] = \mathbb{E}[\lambda^{T,i-}_t]$ and there are on average as many upward than downward jumps, which we interpret as a no statistical arbitrage property.
\\ \\
Define, for any $1 \leq i,j \leq m$,
\begin{align}
    \delta_{ij}^T &:= \psi^T_{j+,i+} - \psi^T_{j-,i+} \\
    \Delta_{ij} &:= \underset{T \to \infty}{\lim} \normOne{\psi^T_{j+,i+}} - \normOne{\psi^T_{j-,i+}} \label{def:Delta_matrix}.
\end{align}
We can make the following remark.
\begin{remark}
Note that for any $1 \leq k \leq m$, defining $\bm{v_k} := \bm{e_{k+}}-\bm{e_{k-}}$ and using \eqref{ass:no_pair_trading_arbitrage} of Assumption \ref{ass:no_arb}, we have
\begin{align*}
    \trans{\bm{\psi^T}} \bm{v_k} &= \trans{\bm{\psi^T}} (\bm{e_{k+}}-\bm{e_{k-}}) \\
    &= (\psi^T_{k+,1+} - \psi^T_{k-,1+}) \bm{e_{1+}} + (\psi^T_{k+,1-} - \psi^T_{k-,1-}) \bm{e_{1-}} + \cdots + (\psi^T_{m+,1+} - \psi^T_{m-,1+}) \bm{e_{m-}} \\
    &= (\psi^T_{k+,1+} - \psi^T_{k-,1+}) \bm{e_{1+}} - (\psi^T_{k+,1+} - \psi^T_{k-,1+}) \bm{e_{1-}} + \cdots + (\psi^T_{m+,1+} - \psi^T_{m-,1+}) \bm{e_{m-}} \\
    &= (\psi^T_{k+,1+} - \psi^T_{k-,1+}) \bm{v_{1}} + \cdots + (\psi^T_{k+,m+} - \psi^T_{k-,m+}) \bm{v_{m}} \\
    &= \delta^T_{k1} \bm{v_1} + \cdots + \delta^T_{kn} \bm{v_m}.
\end{align*}
A sufficient condition for the no pair-trading arbitrage Equation \eqref{ass:no_pair_trading_arbitrage} of Assumption \ref{ass:no_arb} to hold is that, for all $1 \leq i \leq m$, 
$$
\trans{\bm{\phi^T}} \bm{v_i} = \sum_{1 \leq j \leq m} (\trans{\bm{\phi^T}} \bm{v_i} \cdot \bm{v_j})\bm{v_j},
$$
since then we have, for any $1 \leq k \leq m$,
$$
\sum_{1 \leq l \leq m} (\psi^T_{k+,l+} - \psi^T_{k-,l+}) \bm{e_{l+}} - (\psi^T_{k+,l+} - \psi^T_{k-,l+}) \bm{e_{l-}} = \sum_{1 \leq l \leq m} (\psi^T_{k+,l+} - \psi^T_{k-,l+}) \bm{e_{l+}} - (\psi^T_{k+,l-} - \psi^T_{k-,l-}) \bm{e_{l-}}.
$$
In our applications in Section \ref{sec:example} we will use this condition as it is easier to check assumptions on $\bm{\phi}$ than on $\bm{\psi}$.
\end{remark}

\section{Main results}
\label{sec:main_results}

We are now in the position to rigorously state the main results of this paper. We use the processes $\bm{X^T}, \bm{Y^T}$ and $\bm{Z^T}$ defined in the introduction (see Equations \eqref{eq:definition_rescaled_processes_X}, \eqref{eq:definition_rescaled_processes_Y} \eqref{eq:definition_rescaled_processes_Z}) and write 
$$
\bm{O}^{-1} = \begin{pmatrix}
\bm{O_{11}^{(-1)}} & \bm{O_{12}^{(-1)}} \\
\bm{O_{21}^{(-1)}} & \bm{O_{22}^{(-1)}}
\end{pmatrix}, \quad \bm{O} = \begin{pmatrix}
\bm{O_{11}} & \bm{O_{12}} \\
\bm{O_{21}} & \bm{O_{22}}
\end{pmatrix}.
$$
We set
\begin{align*}
    \bm{\Theta^1} &:= \big( \bm{O_{11}} + \bm{O_{12}}(\bm{I} - \integralInf{\bm{C}})^{-1} \integralInf{\bm{B}} \big) \bm{K}^{-1} \\
    \bm{\Theta^2} &:= \big( \bm{O_{21}} + \bm{O_{22}} (\bm{I} - \integralInf{\bm{C}})^{-1} \integralInf{\bm{B}} \big) \bm{K}^{-1} \\
    \bm{\theta}_0 &:= \begin{pmatrix}
            \bm{O^{(-1)}_{11}} & \bm{0} \\
            \bm{0} & \bm{O^{(-1)}_{12}}
            \end{pmatrix} \bm{\mu} \\
    \bm{\Lambda} &:= \dfrac{\alpha}{\Gamma(1 - \alpha)} \bm{K} \bm{M}^{-1}.
\end{align*}
We have the following theorem.
\begin{theorem}
\label{thm:ConvergenceInt}
The sequence $(\bm{X^T}, \bm{Y^T}, \bm{Z^T})$ is $C$-tight for the Skorokhod topology. Furthermore, for every limit point $(\bm{X}, \bm{Y}, \bm{Z})$ of the sequence, there exists a positive process $\bm{V}$ and an $2m$-dimensional Brownian motion $\bm{B}$ such that
\begin{enumerate}[(i)]
    \itemsep0em
    \item $\bm{X_t} = \int_0^t \bm{V_s} ds$, $\bm{Z_t} = \int_0^t \textnormal{diag}(\sqrt{\bm{V_s}})d\bm{B_s}$.
    \item There exists $\bm{\Tilde{V}}$ a process of H\"older regularity $\alpha - 1/2 - \varepsilon$ for any $\varepsilon>0$ such that $\bm{\Theta^1}\bm{\Tilde{V}} = (V^{1}, \cdots, V^{n_c})$, $\bm{\Theta^2}\bm{\Tilde{V}} = (V^{n_c+1}, \cdots, V^{2m})$ and $\bm{\Tilde{V}}$ is solution of the following stochastic Volterra equation:
\begin{equation}
    \label{eq:rough_sde_noF}
    \begin{split}
        \forall t \in [0,1], \bm{\Tilde{V}}_t &= \dfrac{1}{\Gamma(\alpha)}\bm{\Lambda} \int_0^t (t-s)^{\alpha-1}(\bm{\theta}_0 - \bm{\Tilde{V}_s}) ds  \\ &+\dfrac{1}{\Gamma(\alpha)}\bm{\Lambda} \int_0^t (t-s)^{\alpha-1} \bm{O^{(-1)}_{11}}\textnormal{diag}(\sqrt{\bm{\Theta^1} \bm{\Tilde{V}_s}}) d\bm{W^1_s} 
        \\ &+\dfrac{1}{\Gamma(\alpha)}\bm{\Lambda} \int_0^t (t-s)^{\alpha-1} \bm{O^{(-1)}_{12}}\textnormal{diag}(\sqrt{\bm{\Theta^2} \bm{\Tilde{V}_s}}) d\bm{W^2_s},
    \end{split}
\end{equation}
    where $\bm{W^1} := (B^1, \cdots, B^{n_c})$, $\bm{W^2} := (B^{n_c+1}, \cdots, B^{2m})$, $\bm{\Theta^1}$, $\bm{\Theta^2}$, $\bm{O^{(-1)}_{11}}$, $\bm{O^{(-1)}_{12}}, \bm{\theta_0}$ do not depend on the chosen basis.
\end{enumerate}
Finally, any limit point $\bm{P}$ of the rescaled price processes $\bm{P^T}$ satisfies
$$
\bm{P}_t = (\bm{I} + \bm{\Delta}) \trans{\bm{Q}} (\int^t_0 \textnormal{diag}(\sqrt{\bm{V_s}}) d\bm{B}_s + \int_0^t \bm{\mu_s} ds),
$$
where $\bm{\Delta}$ is defined in Equation \eqref{def:Delta_matrix}.
\end{theorem}

Theorem \ref{thm:ConvergenceInt} links multivariate nearly unstable Hawkes processes and multivariate rough volatility. We note that:
\begin{itemize}
    \itemsep0em
    \item The resulting stochastic Volterra equation has non-trivial solutions, as the examples in Section \ref{sec:example} will show.
    \item From a financial perspective, Theorem \ref{thm:ConvergenceInt} shows that the limiting volatility process for a given asset is a sum of different factors. The matrix $\bm{\Delta}$ mixes them and is therefore responsible for correlations between asset prices. Remarks and comments on $\bm{I}+\bm{\Delta}$ are developed in Section \ref{sec:example}.
    \item The theorem implies that adding/removing an asset to/from a market has an impact on the individual volatility of other assets. We can estimate the magnitude of such volatility modifications by calibrating Hawkes processes on price changes.
    \item Since there is a one to one correspondence between the Hurst exponent $H$ and the long memory parameter of the order flow $\alpha$, our model yields the same roughness for all assets. Extensions to allow for different exponents to coexist, for example by introducing an asset-dependent scaling through $\bm{D} = (\alpha_1, \cdots, \alpha_m)$ and studying $\bm{T^{-D}}\bm{\lambda^{T}_{tT}}$, are more intricate. In particular, one needs to use a special function extending the Mittag-Leffler matrix function such that its Laplace transform is of the form $(\bm{I} + \bm{\Lambda t^{D}})^{-1}$.
\end{itemize}

\section{Applications}
\label{sec:example}

In this section, we give examples of processes obtained through Theorem \ref{thm:ConvergenceInt} under different assumptions on the microscopic parameters. The first example highlights the flexibility of our framework and shows that the obtained limit in Theorem \ref{thm:ConvergenceInt} is non-trivial. We then study the influence of microscopic parameters on the limiting price and volatility processes when modeling two assets. Finally, we model many different assets to reproduce realistic high-dimensional correlation matrices.

\subsection{An example of non-trivial volatility process obtained through Theorem \ref{thm:ConvergenceInt}}

Before presenting some truly relevant results for finance, we develop an example demonstrating that the solutions to the Volterra equations of the form of Equation \eqref{eq:rough_sde_noF} are non-trivial. The structure of our Volterra equations is close to those studied in \cite{Jaber2017AffineProcesses}, which proves existence and uniqueness of affine Volterra equations. In particular, this paper covers Volterra equations of the following type, for $\alpha \in (1/2,1)$:
$$
\bm{X_t} = \bm{X_0} + \int_0^t (t-s)^{\alpha - 1} \bm{b(X_s)} ds + \int_0^t (t-s)^{\alpha - 1} \bm{\sigma(X_s)} d\bm{B_s},
$$
where $\bm{b} \colon \mathbb{R} \to \mathbb{R}^n$ and $\bm{\sigma} \colon \mathbb{R} \to \mathcal{M}_{n}(\mathbb{R})$ are continuous functions. A key condition required for existence and uniqueness is sublinear growth condition on $\bm{b}$ and $\bm{\sigma}$, that is
\begin{equation}
\label{eq:condition_eduardo}
\lVert \bm{b(x)} \rVert_2 \vee \lVert \bm{\sigma(x)} \rVert_2 \leq c(1+\lVert \bm{x} \rVert_2),
\end{equation}
for some constant $c>0$ where $\lVert \cdot \rVert_2$ is the usual Euclidian norm for vectors and matrices. Thus, this setting covers equations of the type
$$
\bm{X_t} = \bm{X_0} + \int_0^t (t-s)^{\alpha - 1} \bm{b(X_s)} ds + \int_0^t (t-s)^{\alpha - 1} \textnormal{diag}(\sqrt{\bm{X_s}}) d\bm{B_s},
$$
which are a particular case of Theorem \ref{thm:ConvergenceInt}. However, note that Condition \ref{eq:condition_eduardo} fails when $\bm{\sigma(x)} = \bm{\Sigma} \textnormal{diag}(\sqrt{\bm{x}})$ for some non-diagonal matrix $\bm{\Sigma}$. Interestingly, this setting is covered in our approach as illustrated by the following corollary.
\begin{corollary}
\label{cor:correlation_vol}
We can find a microscopic process satisfying the assumptions of Theorem \ref{thm:ConvergenceInt} such that $\bm{V}$ is a non-negative process which satisfies, for any $t$ in [0,1],
\begin{equation*}
    \bm{V}_t = \int_0^t (t-s)^{\alpha - 1} (
    \bm{\theta} -  \bm{G} \bm{V_s}) ds + \int_0^t (t-s)^{\alpha - 1}
    \bm{\Sigma} \textnormal{diag}(\sqrt{\bm{V_s}}) d\bm{B}_s,
\end{equation*}
where $\bm{\theta}$ is a $4$-dimensional vector, $\bm{G}, \bm{\Sigma}$ are $4 \times 4$ non-diagonal matrices and $\bm{B}$ is a $4$-dimensional Brownian motion.
\end{corollary}

Thus, our framework yields non-trivial solutions and leads to interesting new examples of processes. We now focus on building realistic models to discuss the correspondence between the microscopic parameters of the Hawkes kernel and macroscopic quantities such as correlations and volatility.

\subsection{Influence of microscopic properties on the price dynamics of two correlated assets}
Our first model to understand the price formation process focuses on two assets. Let $\mu^1, \mu^2 > 0$, $\alpha \in (1/2,1), \gamma_1,  \gamma_2$ in $[0,1]$, $H^c_{21}, H^a_{21}, H^c_{12},  H^a_{12}$ in $[0,1]$ \footnote{The superscripts $c$ (resp. $a$) stand for continuation (resp. alternation) to describe that after a price move in a given direction, $H^c$ (resp. $H^a$) encodes the tendency to trigger other price moves in the same (resp. opposite) direction will follow.} such that (here $\sqrt{\cdot}$ is the principal square root, so that if $x<0$, $\sqrt{x}=i \sqrt{-x}$):
\begin{align*}
    0 &\leq (H^c_{12}+H^a_{12})(H^c_{21}+H^a_{21}) < 1 \\
    0 &\leq \modulus{1 - (\gamma_1 + \gamma_2) - \sqrt{(H^c_{12}-H^a_{12})(H^c_{21}-H^a_{21}) + (\gamma_1-\gamma_2)^2}} < 1 \\
    0 &\leq \modulus{1 - (\gamma_1 + \gamma_2) + \sqrt{(H^c_{12}-H^a_{12})(H^c_{21}-H^a_{21}) + (\gamma_1-\gamma_2)^2}} < 1.
\end{align*}
We now have to choose a kernel which satisfies the different assumptions of Section \ref{sec:assumptions} to model the interactions between our two assets. Theorem \ref{thm:ConvergenceInt} states that the only relevant parameters for the macroscopic price are $\bm{K}$ and $\bm{M}$. For simplicity we choose the kernel such that $\bm{M} = \alpha \bm{I}$. This leads us to define, for $t \geq 0$,
\begin{align*}
    \phi^T_1(t) &=  (1-\gamma_1) \alpha (1 - T^{-\alpha})\mathbb{1}_{t \geq 1} t^{-(\alpha + 1)} &\phi^{T,c}_{21}(t) = \alpha T^{-\alpha} H^c_{21} \mathbb{1}_{t \geq 1} t^{-(\alpha + 1)} \\
    \phi^T_2(t) &=  \gamma_1 \alpha (1 - T^{-\alpha})\mathbb{1}_{t \geq 1} t^{-(\alpha + 1)} &\phi^{T,a}_{21}(t) = \alpha T^{-\alpha} H^a_{21} \mathbb{1}_{t \geq 1} t^{-(\alpha + 1)} \\
    \tilde{\phi}^T_1(t) &=  (1- \gamma_2) \alpha (1 - T^{-\alpha})\mathbb{1}_{t \geq 1} t^{-(\alpha + 1)}  &\phi^{T,c}_{12}(t) = \alpha T^{-\alpha} H^c_{12} \mathbb{1}_{t \geq 1} t^{-(\alpha + 1)} \\
    \tilde{\phi}^T_2(t) &=   \gamma_2  \alpha (1 - T^{-\alpha})\mathbb{1}_{t \geq 1} t^{-(\alpha + 1)} &\phi^{T,a}_{12}(t) = \alpha T^{-\alpha} H^a_{12} \mathbb{1}_{t \geq 1} t^{-(\alpha + 1)}.
\end{align*}
For a realistic model, we impose the exogenous source of upward and downward price moves to be equal: $\mu^{1+} = \mu^{1-}$ and $\mu^{2+} = \mu^{2-}$. Thus, the sequence of baselines and kernels are chosen as
$$
\bm{\mu^T} = T^{\alpha - 1}
\begin{pmatrix}
\mu^{1} \\
\mu^{1} \\
\mu^{2} \\
\mu^{2}
\end{pmatrix}, \quad
\bm{\phi^T} = \begin{pmatrix}
\phi_1^T & \phi_2^T & \phi^{T,c}_{12} &  \phi^{T,a}_{12} \\
\phi_2^T & \phi_1^T  & \phi^{T,a}_{12} & \phi^{T,c}_{12} \\
\phi^{T,c}_{21} & \phi^{T,a}_{21} & \tilde{\phi}_{1}^T & \tilde{\phi}_{2}^T \\
\phi^{T,a}_{21} & \phi^{T,c}_{21} &  \tilde{\phi}_{2}^T &   \tilde{\phi}_1^T \\
\end{pmatrix}.
$$
Applying theorem \ref{thm:ConvergenceInt} yields the following result.
\begin{corollary}
\label{cor:example_2D}

Consider any limit point $\bm{P}$ of $\bm{P^T}$. Under the above assumptions, it satisfies
\begin{equation}
    \label{eq:prices_cross}
    \bm{P_t} = \dfrac{\sqrt{2}}{4 \gamma_1 \gamma_2 - (H_{12}^c - H_{12}^a)(H_{21}^c - H_{21}^a)}\begin{pmatrix}
   2 \gamma_2 & H_{21}^c - H_{21}^a \\
   H_{12}^c - H_{12}^a & 2 \gamma_1
    \end{pmatrix}
    \int_0^t
    \begin{pmatrix}
    \sqrt{V^1_s} dW^1_s \\
    \sqrt{V^2_s} dW^2_s
    \end{pmatrix},
\end{equation}
with
\begin{align}
 \label{eq:vol_cross}
    \begin{pmatrix}
V^1_t \\
V^2_t
\end{pmatrix} &=  \dfrac{\alpha}{\Gamma(\alpha)\Gamma(1 - \alpha)} \int_0^t (t-s)^{\alpha - 1} \left( \begin{pmatrix}
    \mu^1 \\
    \mu^2
    \end{pmatrix} - \dfrac{1}{1 - (H^c_{12}+H^a_{12})(H^c_{21}+H^a_{21})} \begin{pmatrix}
1 & H^c_{21}+H^a_{21} \\
H^c_{12}+H^a_{12} & 1
\end{pmatrix} \begin{pmatrix}
V^1_s \\
V^2_s
\end{pmatrix} \right)ds \nonumber \\
& + \sqrt{2}  \dfrac{\alpha}{\Gamma(\alpha)\Gamma(1 - \alpha)} \int_0^t (t-s)^{\alpha - 1} \begin{pmatrix}
    \sqrt{V^1_s} dZ^1_s \\
    \sqrt{V^2_s} dZ^2_s
    \end{pmatrix},
\end{align}
\end{corollary}
where $\bm{W}$ and $\bm{Z}$ are bi-dimensional independent Brownian motions. This model helps us understand how microscopic parameters drive the price formation process to generate a macroscopic price and volatility. We begin our remarks with some definitions. 
\\ \\
We call \textit{momentum} the trend (i.e., the imbalance between the number of upward and downward jumps) created by jumps of one asset on itself . The opposite effect is referred to as \textit{ mean-reversion}. For example, the parameter $\gamma_1$ controls the intensity of self-induced bid-ask bounce on Asset 1: when $\gamma_1$ close to zero corresponds to a strong momentum while $\gamma_1$ close to one corresponds to a strong mean-reversion.
\\ \\
We call \textit{cross-asset momentum} the trend created by jumps of one asset on another. For example, cross-asset momentum from Asset 2 to Asset 1 (resp. Asset 1 to Asset 2) appears via $H^c_{21} - H^a_{21}$ (resp. $H^c_{12} - H^a_{12}$): when both $H^c_{21} - H^a_{21}$ and $H^c_{12} - H^a_{12}$ are nill, the prices of Asset 1 and Asset 2 are uncorrelated. We now turn to comments on the volatility process.
\\ \\
Because of its role in the single-asset case, we refer to $\bm{V}$ as the \textit{fundamental variance}: for example $V^1$ is the fundamental variance of Asset 1. The equation satisfied by $\bm{V}$ only depends on the sum of the feedback effects between each asset through $H^c_{12} + H^a_{12}$: from a volatility viewpoint, upward and downward jumps have the same impact. Furthermore, we can compute the expected fundamental variance using Mittag-Leffler functions (see Section \ref{sec:proofs}). 
\\ \\
Mean-reversion drives down volatility while cross-asset momentum increases it. Indeed, computing $\mathbb{E}[(P^{1}_t)^2]$ for example we get:
\begin{equation*}
    \mathbb{E}[(P^1_t)^2]=2 \dfrac{4\gamma_2^2 \int_0^t \mathbb{E}[V^1_s] ds + (H^{c}_{12} - H^{a}_{12})(H^{c}_{21} - H^{a}_{21}) \int_0^t \mathbb{E}[V^2_s] ds }{[4 \gamma_1 \gamma_2 - (H_{12}^c - H_{12}^a)(H_{21}^c - H_{21}^a)]^2}.
\end{equation*} 
In particular, increasing $\gamma_1$ or $\gamma_2$ does not change $\bm{V}$ but reduces $\mathbb{E}[(P^1_t)^2]$. This example may be particularly relevant to understand the contribution of Asset 2 to the volatility of Asset 1 through calibration to market data since if Asset 2 were removed from the market, we would have $\mathbb{E}[(P^1_t)^2]=\dfrac{1}{2 \gamma_1}$. Focusing now on the price formation process, we see that it results from a combination of momentum, mean-reversion and cross-asset momentum. We illustrate this in two extreme cases: when there is no cross-asset momentum and when cross-asset momentum is strong.
\begin{itemize}
    \item When there is no cross-asset momentum (i.e. $H^c_{12}=H^a_{12}$ and $ H^c_{21}=H^a_{21}$) at the microscopic scale a price move on Asset 2 has the same impact on the intensity of upward and downward price moves of Asset 1. Thus the difference between the expected number of upward and downward jumps does not change after a price move on Asset 2: the expected microscopic price of Asset 1 is unaffected and price moves of Asset 2 generate no trend on Asset 1. This results in macroscopic prices being uncorrelated (see Equation \eqref{eq:prices_cross}).
    \item On the other hand, when cross-asset momentum is strong (i.e. $(H^{c}_{12} - H^{a}_{12})(H^{c}_{21} - H^{a}_{21}) \approx 4 \gamma_1 \gamma_2$, for example if $H^{c}_{12} - H^{a}_{12} = 2\gamma_1\sqrt{1 - \epsilon}, H^{c}_{12} - H^{a}_{12} = 2\gamma_2\sqrt{1 - \epsilon}$ for some small $\epsilon > 0$), at the microscopic scale, a price move on Asset 2 significantly increases the probability of a future price move of Asset 1 in the same direction (and vice-versa). In this context we have
    $$
    \bm{\Delta + I} = \dfrac{1}{2\gamma_1 \gamma_2 \epsilon}\begin{pmatrix}
        \gamma_2 & \gamma_2\sqrt{1 - \epsilon} \\
       \gamma_1\sqrt{1 - \epsilon} &  \gamma_1
        \end{pmatrix}.
    $$
    Using Equation \eqref{eq:prices_cross} we can check that $\dfrac{\mathbb{E}[P^1_t P^2_t]}{\sqrt{\mathbb{E}[(P^1_t)^2] \mathbb{E}[(P^2_t)^2]}} \underset{\epsilon \to 0}{\to} 1$ and prices evolve in unison.
\end{itemize} 
This example underlines that in our approach (thanks to our no-arbitrage constraint) microscopic features transfer to macroscopic properties in an intuitive way.

\subsection{Reproducing realistic correlation matrices of large number of assets using microscopic properties}

It is well-known that the correlation matrix of stocks has few large eigenvalues outside of a "bulk" of eigenvalues attributable to noise (see for example \cite{Laloux1999NoiseMatrices}). The largest eigenvalue is referred to as the market mode (because the associated eigenvector places a roughly equal weight on each asset) and is much larger than other eigenvalues. Other significant eigenvalues can be related to the presence of sectors: groups of companies with similar characteristics.
\\ \\
How can we provide microstructural foundations for this stylised fact? The large eigenvalue associated to the market mode implies that, in a first approximation, stock prices move together: a price increase on one asset is likely followed by a price increase on all other assets. Translating this in our framework, an upward (resp. downward) jump on a given asset increases the probability of an upward (resp. downward) jump on all other assets. We further expect that an upward price move on an asset increases this probability much more on an asset from the same sector than on an unrelated one.
\\ \\
The above remarks lead us to consider a model where:
\begin{itemize}
    \itemsep0em
    \item All stocks share some fundamental high-frequency properties by having similar self-excitement parameters in the kernel.
    \item Stocks have a stronger influence on price changes of stocks within the same sector.
    \item Within the same sector, all stocks have the same microscopic parameters.
\end{itemize}
The technical details of our setting are presented in Appendix \ref{sec:proof_example_nd} and we only provide here essential elements to understand the framework. Let $\mu^1, \dots, \mu^m > 0$ be the baselines of each asset. Using the same notations as before, take $\gamma$ in $[0,1]$, $\alpha$ in $(1/2,1)$ and $H^c, H^a > 0$. We consider $R > 0$ different sectors, Sector $r$ having $m_r$ stocks. For a pair of stocks which we dub $1,2$ to make an analogy with the previous example, we have that:
\begin{itemize}
    \item The self excitement parameters are equal: $\gamma_1 = \gamma_2 = \gamma$ where $\gamma$ is the same for all stocks.
    \item If Stock 1 and Stock 2 do not belong to the same sector, $H_{21}^c = H_{12}^c = H^c$, $H_{21}^a = H_{12}^a = H^a$ where $H^c,H^a$ are the same for all stocks.
    \item If Stock 1 and Stock 2 belong to the same sector $r$, $H_{21}^c = H_{12}^c = H^c + H^c_r$, $H_{21}^a = H_{12}^a = H^a + H^a_r$ where $H^c_r,H^a_r$ are the same for all stocks belonging to sector $r$.
\end{itemize}
The asymptotic framework is built as in the previous example, with the details given in the proof of Corollary \ref{cor:example_nD} in Appendix \ref{sec:proof_example_nd}. We write $i_r := m_0 + m_1 + \cdots + m_{r-1}$ for $1 \leq r \leq R$ (with convention $m_0=1$) so that stocks from sector $r$ are indexed between $i_{r}$ and $i_{r+1}$ excluded and define the following vectors
\begin{align*}
    \bm{w} &:= \dfrac{1}{\sqrt{m}}(\bm{e_1} + \cdots + \bm{e_{m}}) \\
    \bm{w_r} &:= \dfrac{1}{\sqrt{m_r}}\sum_{i_r \leq i < i_{r+1}}\bm{e_{i}} \\
    \bm{\theta} &:= \sum_{1 \leq i \leq m} \mu^i \bm{e_i}.
\end{align*} 
We consider an asymptotic framework where the number of assets will eventually grow to infinity. As will become clear in the proof, the only non-trivial regime appears when $H^{c},H^{a},H^{c}_r,H^{a}_r \underset{m \to \infty}{=} \mathcal{O}(m^{-1})$.  Thus we assume that $m H^{c},mH^{a},mH^{c}_r,mH^{a}_r$ converge to $\bar{H}^{c},\bar{H}^{a},\bar{H}^{c}_r,\bar{H}^{a}_r$ as $m$ tends to infinity. We also assume that the proportion of stocks in a given sector relative to the total number of stocks does not vanish: for each $1 \leq r \leq R$, $\frac{m_r}{m} \underset{m \to \infty}{\to} \eta_r > 0$. Define the following constants which will appear in the price and volatility processes: $\lambda^{+} := \bar{H}^{c} + \bar{H}^{a}, \lambda_r^{+} := \bar{H}^{c}_r + \bar{H}^{c}_r$, $\lambda^{-} := \bar{H}^{c} - \bar{H}^{a}, \lambda^{-}_r := \bar{H}^{c}_r - \bar{H}^{a}_r$. Applying Theorem \ref{thm:ConvergenceInt} yields the following result.
\begin{corollary}
\label{cor:example_nD}
Consider any limit point $\bm{P}$ of $\bm{P^T}$. Under the above assumptions, it satisfies:
\begin{align*}
    \bm{P}_t &= \sqrt{2} \bm{\Sigma}_{\varepsilon} \int_{0}^{t}\textnormal{diag}(\sqrt{\bm{V}_s}) d\bm{W_s},
\end{align*}
where $\bm{W}$ is a Brownian motion, $\bm{\Sigma}_{\varepsilon} := (2 \gamma \bm{I} - \lambda^{-} \bm{v} \trans{\bm{v}} - \sum_{1 \leq r \leq R} \eta_r \lambda^{-}_r \bm{v_r} \trans{\bm{v_r}} + \bm{\varepsilon})^{-1}$ with $\bm{\epsilon}$ a deterministic $m \times m$ matrix such that $\spectralRadius{\bm{\epsilon}} \underset{m \to \infty}{=} o(m^{-1})$ and $\bm{V}$ satisfies the stochastic Volterra equation
\begin{align*}
    \bm{V_t} &=  \dfrac{\alpha}{\Gamma(\alpha)\Gamma(1 - \alpha)} \int_0^t (t-s)^{\alpha-1} (\bm{\theta} - \bm{\mathcal{V}}_{\epsilon} \bm{V_s})ds +  \dfrac{\sqrt{2} \alpha}{\Gamma(\alpha)\Gamma(1 - \alpha)} \int_0^t (t-s)^{\alpha-1}  \textnormal{diag}(\sqrt{\bm{V_s}})d\bm{Z_s},
\end{align*}
with $\bm{Z}$ a Brownian motion independent from $\bm{W}$ and $\bm{\mathcal{V}}_{\epsilon} :=  \left( \bm{I} - \lambda^{+} \bm{v} \trans{\bm{v}} - \sum_{1 \leq r \leq R} \eta_r \lambda^{+}_r \bm{v_r} \trans{\bm{v_r}} + \bm{\epsilon} \right)^{-1}$ where $\bm{\varepsilon}$ is a deterministic $m \times m$ matrix such that $\spectralRadius{\bm{\varepsilon}} \underset{m \to \infty}{=} o(m^{-1})$.
\end{corollary}

Under the previous corollary, writing $\propto$ for equality up to a multiplicative constant, the expected fundamental variance can be written using the cumulative Mittag-Leffler function
\begin{equation*}
     \mathbb{E}[\bm{V_t}] \propto \bm{F^{\alpha, \bm{\mathcal{V}}_{\epsilon}}(t)} \bm{\theta}.
\end{equation*}
Since $\spectralRadius{\bm{\epsilon}} \underset{m \to \infty}{=} o(m^{-1})$, we neglect it in further comments and use the approximation $\bm{\mathcal{V}}_{\epsilon} \approx \bm{\mathcal{V}}_{0}$. Writing $\xi$ for the largest eigenvalue of $\bm{\mathcal{V}}_{0}$ and neglecting other eigenvalues (which is reasonable if $\lambda^+ + \sum_{1 \leq r \leq R} \eta_r \lambda^{+}_r \approx 1$) and $\bm{z}$ for the associated eigenvector, using the definition of the Mittag-Leffler function (see Definition \ref{def:mittag_leffler_matrices} in Appendix \ref{sec:fractional_operators}), we have
\begin{equation*}
     \mathbb{E}[\bm{V_t}] \propto F^{\alpha, \xi}(t) (\trans{\bm{\theta}} \bm{z}) \bm{z}.
\end{equation*}
In the further approximation that $\eta_r \lambda^+_r$ is independent $r$, we have $\bm{z} \propto (1, \cdots, 1)$  and
\begin{align*}
\mathbb{E}[\bm{P_t} \trans{\bm{P_t}}] &\propto \bm{\Sigma_{\varepsilon}} \textnormal{diag}(\mathbb{E}[\bm{V_t}]) \trans{\bm{\Sigma_{\varepsilon}}} \\
&\propto \bm{\Sigma_{\varepsilon}} \textnormal{diag}(\bm{z}) \trans{\bm{\Sigma_{\varepsilon}}}  \\
& \propto \bm{\Sigma_{\varepsilon}} \trans{\bm{\Sigma_{\varepsilon}}} \propto \bm{\Sigma_{\varepsilon}}^2.    
\end{align*}
Therefore the eigenvectors of $\mathbb{E}[\bm{P_t} \trans{\bm{P_t}}]$ are those of $\bm{\Sigma_{\varepsilon}}$. As $\spectralRadius{\bm{\varepsilon}} \underset{m \to \infty}{=} o(m^{-1})$, we neglect it in further comments and use the approximation $\bm{\Sigma_{\varepsilon}} \approx \bm{\Sigma_{0}}$. When $\lambda^{-} +  \sum_{1 \leq r \leq R} \eta_r \lambda^{-}_r \approx 2\gamma$, $\bm{\Sigma_{0}}$ has one large eigenvalue followed by $R-1$ smaller eigenvalues and much smaller eigenvalues. This is consistent with stylised facts of high-dimensional stock correlation matrices and we have thus built a microscopic model to explain the macroscopic structure of correlation matrices.
\\ \\
The conditions $\lambda^{-} +  \sum_{1 \leq r \leq R} \eta_r \lambda^{-}_r \approx 1$ and $\lambda^+ + \sum_{1 \leq r \leq R} \eta_r \lambda^{+}_r \approx 1$  correspond to the parameters being close to the point where all directions are critical: when $\lambda^{-} +  \sum_{1 \leq r \leq R} \eta_r \lambda^{-}_r \approx 2\gamma$ or $\lambda^{-} +  \sum_{1 \leq r \leq R} \eta_r \lambda^{-}_r \approx 1$, the spectral radius of $\norm{\bm{C}}$ is equal to one and we cannot split the kernel into a critical and a non-critical component.
\\ \\
It would be interesting to study other implications of this model. In particular, we believe that encoding a negative price/volatility correlation into the microscopic parameters could explain the so-called index leverage effect (see \cite{Reigneron2011PrincipalEffect} for a definition and empirical analysis of this stylised fact).

\section{Proof of Theorem \ref{thm:ConvergenceInt}}
\label{sec:proofs}
We split the proof into four steps. Our approach is inspired by \cite{ElEuch2018TheVolatility.,Jaisson2015LimitProcesses, Jaisson2016RoughProcesses}. First, we show that the sequence $(\bm{X^T}, \bm{Y^T}, \bm{Z^T})$ is $C$-tight. Second, we use tightness and representation theorems to find equations satisfied by any limit point $(\bm{X}, \bm{Y}, \bm{Z})$ of $(\bm{X^T}, \bm{Y^T}, \bm{Z^T})$. Third, properties of the Mittag-Leffler function enable us to prove Equation \eqref{eq:rough_sde_noF}. Fourth and finally, we derive the equation satisfied by any limit point $\bm{P}$ of $\bm{P^T}$.

\subsection*{Preliminary lemmas}

We start with lemmas that will be useful in the proofs. Lemma A.1 from \cite{ElEuch2018TheVolatility.} yields
\begin{equation}
    \label{eq:intensity_wiener_hopf}
    \dfrac{1}{T^{\alpha}}\bm{\lambda^T_{tT}} = \dfrac{\bm{\mu^T_{tT}}}{T^{\alpha}} + \dfrac{1}{T^{\alpha}} \int_0^{tT} \bm{\psi^T(tT-s)} \bm{\mu^T_{s}} ds + \dfrac{1}{T^{\alpha}} \int_0^{tT} \bm{\psi^T(tT-s)} d\bm{M^T_s}.
\end{equation}
Thus to investigate the limit of $\dfrac{1}{T^{\alpha}}\bm{\lambda^T_{ \cdot T}}$ we need to study $\dfrac{1}{T^{\alpha}} \bm{\psi^T(T \cdot)}$, which we will do through its Laplace transform. Given a $L^1(\mathbb{R}_{+})$ function $f$, we write its Laplace transform $\hat{f}(t) := \int_0^{\infty} f(x)e^{-tx} dx$, for $t \geq 0$ (and similarly for matrix-valued functions $\bm{F} = (F_{ij})$ where each $F_{ij} \in L^1(\mathbb{R}_{+})$). Remark that $\widehat{f^{*k}} = \hat{f}^{k}$, where $*k$ is the convolution product iterated $k$ times. The following lemma holds.
\begin{lemma}
\label{lemma:convergence_laplace_transform}
We have the following convergence for any $t \geq 0$:
\begin{equation} 
\label{eq:convergence_laplace_transform}
T^{-\alpha} \bm{\widehat{\psi^T(T \cdot)}(t)} \underset{T \to \infty}{\to}
 \bm{O} \begin{pmatrix}
     \left[ \dfrac{\Gamma(1 - \alpha)}{\alpha} t^{\alpha}\bm{M} + \bm{K} \right]^{-1} & 0 \\
    (\bm{I} - \norm{\bm{C}})^{-1} \norm{\bm{B}} \left[ \dfrac{\Gamma(1 - \alpha)}{\alpha} t^{\alpha}\bm{M} + \bm{K} \right]^{-1} &  0
    \end{pmatrix} \bm{O}^{-1},
\end{equation}
where $\bm{K}$ and $\bm{M}$ are defined in Equation \eqref{eq:convergence_rates_AT} and \eqref{eq:convergence_rates_A}. 
\end{lemma}

\begin{proof}
Define $\bm{\varphi^T} := \bm{O}^{-1}\hat{\bm{\phi}}^T \bm{O}$. Then
\begin{align*}
\bm{\hat{\psi}^T(t)} & =
\sum_{k \geq 1} \bm{\hat{\phi}^{T,*k}} = \bm{O} (\bm{\bm{I}} - \bm{\hat{\varphi}^T})^{-1} \bm{\hat{\varphi}}^T \bm{O}^{-1}.
\end{align*}

We can use the shape of $\bm{\varphi^T}$ and matrix block inversion to rewrite this expression. Doing so, we find
\begin{align*}
\bm{\hat{\psi}^T(t)} & =
     \bm{O} \begin{pmatrix}
    (\bm{\bm{I}} - \bm{\hat{A}^T(t))}^{-1} \bm{\hat{A}^T(t)} & 0 \\
    (\bm{\bm{I}} - \bm{\hat{C}^T(t)})^{-1} \bm{\hat{B}^T(t)} (\bm{\bm{I}} - \bm{\hat{A}^T(t)})^{-1}\bm{\hat{A}^T(t)} -   (\bm{\bm{I}} - \bm{\hat{C}^T(t)})^{-1}\hat{\bm{B}}^T(t) &  (\bm{\bm{I}} - \bm{\hat{C}^T(t)})^{-1} \bm{\hat{C}^T(t)}
    \end{pmatrix} \bm{O}^{-1}.
\end{align*}
To derive the limiting process, we use Equations \eqref{eq:convergence_rates_AT} and \eqref{eq:convergence_rates_A}. Using integration by parts and a Tauberian theorem as in \cite{ElEuch2018TheVolatility., Jaisson2016RoughProcesses}, we have
\begin{align*}
    \bm{\norm{A^T}} - \bm{\hat{A}^T(t/T)} &\underset{T \to \infty}{=} \dfrac{\Gamma(1 - \alpha)}{\alpha} t^{\alpha}\bm{M} T^{-\alpha} + o(T^{-\alpha}) \\
    \bm{I}- \norm{\bm{A^T}} &\underset{T \to \infty}{=} \bm{K} T^{-\alpha} + o(T^{-\alpha}).
\end{align*}
Therefore 
\begin{align*}
    T(\bm{I} - \bm{\hat{A}^T(t/T)}) &=  T(\norm{\bm{A^T}} - \bm{\hat{A}^T(t/T))} + T(\bm{I}- \norm{\bm{A}^T}) \\
    & \underset{T \to \infty}{=} \left[ \dfrac{\Gamma(1 - \alpha)}{\alpha} t^{\alpha}\bm{M} +  \bm{K} \right] T^{1-\alpha}+ o(T^{1-\alpha}).
\end{align*}
Consequently
\begin{align*}
    T^{\alpha - 1} T(\bm{I} - \bm{\hat{A}^T(t/T)}) &\underset{T \to \infty}{=} \dfrac{\Gamma(1 - \alpha)}{\alpha} t^{\alpha}\bm{M} +  \bm{K} + o(1).
\end{align*}
By Assumption \ref{ass:StructureHawkes} $\bm{M}$ is invertible and $\bm{K}\bm{M}^{-1}$ has strictly positive eigenvalues. Thus $\bm{M}t + \bm{K} = (\bm{K}\bm{M}^{-1}+t\bm{I})\bm{M}$ is invertible for any $t \geq 0$. The Laplace transform of $T^{-\alpha} \bm{\psi^T(T \cdot)}$ being  $T^{1-\alpha} \bm{\widehat{\psi}^T(\cdot/T)}$, we have proved for any $t \geq 0$,
\begin{align*}
\widehat{T^{-\alpha} \bm{\psi^T(T \cdot)}}(t) \underset{T \to \infty}{\to} \bm{O} \begin{pmatrix}
    \left[ \dfrac{\Gamma(1 - \alpha)}{\alpha} t^{\alpha}\bm{M} + \bm{K} \right]^{-1} & 0 \\
    (\bm{I} - \norm{\bm{C}})^{-1} \norm{\bm{B}} \left[ \dfrac{\Gamma(1 - \alpha)}{\alpha} t^{\alpha}\bm{M} + \bm{K} \right]^{-1} &  0
    \end{pmatrix} \bm{O}^{-1}.
\end{align*}

\end{proof}

We show in the technical appendix that the inverse Laplace transform of $\bm{\Lambda} (t^{\alpha} \bm{I} + \bm{\Lambda})^{-1}$, where $\bm{\Lambda} \in \mathcal{\bm{M}}_n(\mathbb{R})$ has positive eigenvalues, is a simple extension of the Mittag-Leffler density function to matrices (see Definition \ref{def:mittag_leffler_matrices} in the appendix) denoted by $\bm{f^{\alpha, \Lambda}}$. Thus we define for any $t \in [0,1]$
\begin{equation}
    \label{eq:definition_f_psi}
    \bm{f(t)} := \bm{O} \begin{pmatrix}
    \bm{K}^{-1} \bm{f^{\alpha,\dfrac{\alpha}{\Gamma(1 - \alpha)}K \bm{M}^{-1}}}  & \bm{0} \\
    (\bm{I} - \norm{\bm{C}})^{-1} \norm{\bm{B}} \bm{K}^{-1} \bm{f^{\alpha,\dfrac{\alpha}{\Gamma(1 - \alpha)}K \bm{M}^{-1}}} &  \bm{0}
    \end{pmatrix}\bm{O}^{-1}.
\end{equation}
The following lemma shows the weak convergence of $\bm{\psi^T}$ towards $\bm{f}$.

\begin{lemma}
\label{lemma:convergence_F_T}
For any bounded measurable function $g$ and $1 \leq i,j \leq n$
\begin{align*}
\int_{[0,1]} g(x) T^{-\alpha} \psi^T_{ij}(Tx) dx & \underset{T \to \infty}{\to} \int_{[0,1]} g(x) f_{ij}(x) dx.
\end{align*}

\end{lemma}

\begin{proof}
First note that when $\normOne{f_{ij}} = 0$ (which implies $f_{ij} = 0$), using Equation \eqref{eq:convergence_laplace_transform} with $t=0$ we have
$$
\normOne{T^{1-\alpha} \psi^T_{ij}} \underset{T \to \infty}{\to} \normOne{f_{ij}} = 0,
$$
which implies, since $1 - \alpha \geq 0$,
$$
\normOne{\psi^T_{ij}} \underset{T \to \infty}{\to} 0.
$$

Therefore,  as $\psi^T_{ij} \geq 0$, for any bounded measurable function $g$
\begin{align*}
\Big \lvert \int_{[0,1]} g(x) T^{-\alpha} \psi^T_{ij}(Tx) dx \Big \rvert \leq c \int_{[0,1]} T^{-\alpha} \psi^T_{ij}(Tx) dx \leq c \normOne{T^{1-\alpha} \psi^T_{ij}},
\end{align*}
and the result holds. Assume now that $\normOne{f_{ij}} > 0$. It will be convenient for us to proceed with random variables, so define 
$$\rho^T_{ij} := \dfrac{T^{-\alpha} \psi^T_{ij}(T \cdot)}{\normOne{T^{1-\alpha}\psi^T_{ij}}}.$$

We can view $\rho^T_{ij}$ as the density of a random variable taking values in $[0,1]$, say $S$. Lemma \ref{lemma:convergence_laplace_transform} gives the convergence of the characteristic functions of $S$ towards 
$$\hat{\rho}_{ij} := \dfrac{\hat{f}_{ij}}{\normOne{f_{ij}}}.$$
Since $\rho_{ij}$ is continuous (as $\psi^T_{ij}$ is continuous), Levy's continuity theorem guarantees that $\rho^T_{ij}$ converges weakly towards $\rho_{ij}$. Therefore for any bounded measurable function $g$
\begin{align*}
\int_{[0,1]} g(x) \rho^T_{ij}(x) dx &  \underset{T \to \infty }{\to} \int_{[0,1]} g(x) \rho_{ij}(x) dx \\
\int_{[0,1]} g(x) \dfrac{T^{-\alpha} \psi^T_{ij}(Tx)}{\normOne{T^{1-\alpha} \psi^T_{ij}}} dx & \underset{T \to \infty }{\to} \int_{[0,1]} g(x) \dfrac{f_{ij}(x)}{\normOne{f_{ij}}} dx.
\end{align*}
Equation \eqref{eq:convergence_laplace_transform} implies $\normOne{T^{1-\alpha} \psi^T_{ij}} \underset{T \to \infty}{\to} \normOne{f_{ij}}$, so that together with the above we have
$$
\int_{[0,1]} g(x) T^{-\alpha} \psi^T_{ij}(Tx) dx  \underset{T \to \infty }{\to} \int_{[0,1]} g(x) f_{ij}(x) dx.
$$
\end{proof}

We introduce the cumulative functions
\begin{align*}
    \bm{F^T(t)} &= \int_0^t T^{-\alpha} \bm{\psi^T(Ts)} ds \\
    \bm{F(t)} &= \int_0^t \bm{f(s)} ds.
\end{align*}
We have just shown in particular that $\bm{F^T}$ converges pointwise towards $\bm{F}$ and therefore, by Dini's theorem, converges uniformly towards $\bm{F}$.

\label{sec:proof1}
\subsection{Step 1: $C$-tightness of $(X^T, Y^T, Z^T)$}

Recall the definition of the rescaled processes:
\begin{align*}
    \bm{X^T_t} &:= \dfrac{1}{T^{2\alpha}} \bm{N^T_{tT}} \\
    \bm{\bm{Y}^T_t} &:= \dfrac{1}{T^{2\alpha}} \int_0^{tT} \bm{\lambda_s} ds  \\
    \bm{\bm{Z}^T_t} &:= T^{\alpha} (\bm{X^T_t} - \bm{\bm{Y}^T_t}) = \dfrac{1}{T^{\alpha}} \bm{M^T_{tT}}.
\end{align*}
As in \cite{ElEuch2018TheVolatility.} and \cite{Jaisson2016RoughProcesses} we show that the limiting processes of $\bm{X^T}$ and $\bm{\bm{Y}^T}$ are the same and that the limiting process of $\bm{\bm{Z}^T}$ is the quadratic variation of the limiting process of $\bm{X^T}$. We have the following proposition:

\begin{proposition}[C-tightness of $(\bm{X^T}, \bm{\bm{Y}^T}, \bm{\bm{Z}^T})$]
\label{prop:c_tightness_sequence}

The sequence $(\bm{X^T}, \bm{\bm{Y}^T}, \bm{\bm{Z}^T})$ is C-tight and if $(\bm{X},\bm{Z})$ is a possible limit point of $(\bm{X^T}, \bm{\bm{Z}^T})$, then $\bm{Z}$ is a continuous martingale with $[\bm{Z},\bm{Z}] = \textnormal{diag}(\bm{X})$. 
Furthermore, we have the convergence in probability
\begin{equation*}
    \underset{t \in [0,1]}{\sup} \normTwo{\bm{\bm{Y}^T_t} - \bm{X^T_t}} \overset{\mathbb{P}}{\underset{T \to \infty}{\to}} 0.
\end{equation*}

\end{proposition}

\begin{proof}
The proof is esentially the same as in \cite{ElEuch2018TheVolatility.}, adapting for our structure of Hawkes processes. We have
$$
\bm{\lambda^T_t} = \bm{\mu^T_t} + \int_0^t \bm{\psi^T(t-s)} \bm{\mu^T_s} ds 
+ \int_0^t \bm{\psi^T(t-s)} d\bm{M^T_s},
$$
and therefore
\begin{align*}
    \mathbb{E}[N^T_T] &= \mathbb{E}[\int_0^T \bm{\lambda}_s^T ds] \\
    & = \int_0^T \bm{\mu^T_t} dt + \int_0^T \int_0^t \bm{\psi^T(t-s)} \bm{\mu^T_s} ds dt \leq c T^{2\alpha} \normInf{\bm{\mu}},
\end{align*}
where we used the convergence of $T^{1-\alpha} \bm{\mu}^T_{T \cdot}$ (see Equation \eqref{eq:convergence_rates_mu}) together with the weak convergence of $T^{-\alpha} \bm{\psi^T}(T \cdot)$ (see Lemma \ref{lemma:convergence_F_T}). It follows then that
$$
\mathbb{E}[\bm{X^T_1}] = \mathbb{E}[\bm{\bm{Y}^T_1}] \leq c,
$$
and since the processes are increasing, $\bm{X^T}$ and $\bm{\bm{Y}^T}$ are tight. As the maximum jump size of $\bm{X^T}$ and $\bm{\bm{Y}^T}$ tends to $0$, we have the $C$-tightness of $(\bm{X^T}, \bm{Y^T})$. Since $\bm{N^T}$ is the quadratic variation of $\bm{M^T}$, $(M^{T,i})^2 - N^{T,i}$ is an $L^2$ martingale starting at $0$ and Doob's inequality yields
\begin{align*}
    \sum_{1 \leq i \leq n} \mathbb{E}[ \underset{t \in [0,1]}{\sup} (X^{T,i}_t - Y^{T,i}_t)^{2}] & \leq 4 \sum_{1 \leq i \leq n} \mathbb{E}[(X^{T,i}_1 - Y^{T,i}_1)^{2}] \\
    & \leq 4 T^{-4\alpha} \sum_{1 \leq i \leq n} \mathbb{E}[(M^{T,i}_{T})^{2}] \\
    & \leq 4 T^{-4\alpha} \sum_{1 \leq i \leq n} \mathbb{E}[N^{T,i}_{T}] \\
    & \leq c T^{-2\alpha}.
\end{align*}
Using the same approach as in \cite{ElEuch2018TheVolatility.} we conclude that $\bm{Z}$ is a continuous martingale and $[\bm{Z},\bm{Z}]$ is the limit of $[\bm{\bm{Z}^T},\bm{\bm{Z}^T}]$. 

\end{proof}

\subsection{Step 2: Rewriting of limit points of $(X^T, Y^T, Z^T)$}

By Proposition \ref{prop:c_tightness_sequence}, for any limit point $(\bm{X}, \bm{Y})$ of $(\bm{X^T}, \bm{Y^T})$, we have $\bm{X} = \bm{Y}$ almost surely. We use $\bm{Y^T}$ to derive an equation for $\bm{Y}=\bm{X}$. As $\bm{Y^T} = \dfrac{1}{T^{2\alpha}}\int_0^{tT} \bm{\lambda^T_s} ds$, we first study $\bm{\lambda^T_{sT}}$. Using Equation \eqref{eq:intensity_wiener_hopf} we get
\begin{align*}
    \int_0^t \bm{\lambda^T_s} ds &= \int_0^t \bm{\mu^T_s} ds +  \int_0^t \int_0^u \bm{\psi^T(s-u)}\bm{\mu^T_u} du ds + \int_0^t \bm{\psi^T(t-s)} \bm{M^T_s} ds \\
    &= \int_0^t \bm{\mu^T_s} ds +  \int_0^t \bm{\psi^T(t-s)} \int_0^s \bm{\mu^T_u} du ds + \int_0^t \bm{\psi^T(t-s)} \bm{M^T_s} ds.
\end{align*}
A change variables of leads to
\begin{align*}
    \int_0^{tT} \bm{\lambda}^T_s ds & = \int_0^{tT} \bm{\mu^T_s} ds +  \int_0^{tT} \bm{\psi^T(tT-s)} \int_0^s \bm{\mu^T_u} du ds + \int_0^{tT} \bm{\psi^T(tT-s)} \bm{M^T_s} ds \\
    & = \int_0^{tT} \bm{\mu^T_s} ds +  T \int_0^{t} \bm{\psi^T(tT-sT)} \int_0^{sT} \bm{\mu^T_u} du ds + \int_0^{t} \bm{\psi^T(tT-sT)} \bm{M^T_{sT}} T ds \\
    & = T \int_0^t \bm{\mu^T_{sT}} ds +  T \int_0^{t} \bm{\psi^T(T(t-s))} \int_0^{sT} \bm{\mu^T_u} du ds + T\int_0^{t} \bm{\psi^T(T(t-s))} \bm{M^T_{sT}} ds.
\end{align*}
Therefore
\begin{align}
    T^{2\alpha} \bm{\bm{Y}^T_t} & = T \int_0^t \bm{\mu^T_{sT}} ds +  T \int_0^{t} \bm{\psi^T(T(t-s))}\int_0^{sT} \bm{\mu^T_u} du ds + T\int_0^{t} \bm{\psi^T(T(t-s))} \bm{M^T_{sT}} ds \label{eq:expression_Y_T}\\
    & =: T^{2 \alpha}(\bm{Y^{T,1}_t} + \bm{Y^{T,2}_t} + \bm{Y^{T,3}_t}),
\end{align}
with obvious notations. Thus, to obtain our limit we use the convergence properties of $\bm{F^T}$ which we derived previously. We have the following proposition.
\begin{proposition}
Consider $(\bm{X}, \bm{Z})$ a limit point of $(\bm{X^T}, \bm{\bm{Z}^T}$). Then,
$$
\bm{X_t} = \int_0^{t} \bm{F(t-s)} \bm{\mu_s} ds + \int_0^{t} \bm{F(t-s)}d\bm{Z_s}.
$$
\end{proposition}

\begin{proof}

Let $(\bm{X}, \bm{Y}, \bm{Z})$ be a limit point of $(\bm{X^T}, \bm{Y^T}, \bm{\bm{Z}^T}$). First, since $T^{1-\alpha} \bm{\mu^T_{tT}} \underset{T \to \infty}{\to} \bm{\mu_t}$ (see Equation \eqref{eq:convergence_rates_mu}), $\bm{Y^{T,1}_t}$ converges to $0$ as $T$ tends to infinity. Moving on to $\bm{Y^{T,2}}$, by integration by parts we have
\begin{align*}
    \bm{Y^{T,2}_t} &= \int_0^{t} T^{1-\alpha} \bm{\psi^T(T(t-s))} T^{-\alpha} \int_0^{sT} \bm{\mu^T_u} du ds \\
    & = \left[ \bm{F^T(t-s)} T^{-\alpha} \int_0^{sT} \bm{\mu^T}_{u} du \right]_{0}^{t} + \int_0^{t} \bm{F^T(t-s)} T^{1-\alpha} \bm{\mu^T}_{sT} ds \\
    & = \int_0^{t} \bm{F^T(t-s)}T^{1-\alpha} \bm{\mu^T}_{sT} ds.
\end{align*}
Using Equation \eqref{eq:convergence_rates_mu} again together with the uniform convergence of $\bm{F^T}$ (see Lemma \ref{lemma:convergence_F_T}) we have the convergence
$$
\bm{Y^{T,2}_t} \underset{T \to \infty}{\to} \int_0^{t} \bm{F(t-s)} \bm{\mu_s} ds.
$$
Finally, $\bm{Y^{T,3}_t}$ can be written as
\begin{align*}
    \bm{Y^{T,3}_t} &= T^{1-2\alpha} \int_0^{t} \bm{\psi^T(T(t-s))} \bm{M}^T_{sT} ds = \int_0^{t} \bm{F^T(t-s)} d\bm{Z^T_s} \\
    & = \int_0^{t} \bm{F(t-s)} d\bm{Z_s} + \int_0^{t} \bm{F(t-s)}(d\bm{\bm{Z}^T_s} - d\bm{Z_s}) + \int_0^{t} (\bm{F^T(t-s)} - \bm{F(t-s)}) d\bm{Z^T_s}.
\end{align*}
The Skorokhod representation theorem applied to $(\bm{\bm{Z}^T}, \bm{Z})$ yields the existence of copies in law $(\Tilde{\bm{Z}}^T, \Tilde{\bm{Z}})$, $\Tilde{\bm{Z}}^T$ converging almost surely to $\Tilde{\bm{Z}}$. We proceed with $(\Tilde{\bm{Z}}^T,\Tilde{\bm{Z}})$ and keep previous notations. The stochastic Fubini theorem \cite{Veraar2012TheRevisited} gives, almost surely
$$
\int_0^{t} \bm{F(t-s)}(d\bm{\bm{Z}^T_s} - d\bm{Z_s}) = \int_0^{t} \bm{f(s)} (\bm{Z^T_{t-s}} - \bm{Z_{t-s}}) ds.
$$
From the dominated convergence theorem we obtain the almost sure convergence
$$
\int_0^{t} \bm{f(s)} (\bm{Z^T_{t-s}} - \bm{Z_{t-s}}) ds \underset{T \to \infty}{\to} 0.
$$
Furthermore, since $[\bm{Z^T},\bm{Z^T}] = \textnormal{diag}(\bm{X^T})$ we have
\begin{align*}
\sum_{1 \leq i \leq n} \mathbb{E}\left[ \left( \int_0^{t} (\bm{F^T(t-s)} - \bm{F(t-s)}) d\bm{\bm{Z}^T_s} \right)_{i}^2 \right] & \leq \sum_{1 \leq i,j \leq n} \int_0^{t} (F_{ij}^T(t-s) - F_{ij}(t-s))^2 T^{1-\alpha} \mathbb{E}[\lambda^{T,j}_{sT}] ds.
\end{align*}
Using Equation \eqref{eq:intensity_wiener_hopf} together with Lemma \ref{lemma:convergence_laplace_transform} we can bound $\mathbb{E}[\lambda^{T,j}_{sT}]$ independently of $T$ and
\begin{align*}
\sum_{1 \leq i \leq n} \mathbb{E}\left[ \left( \int_0^{t} (\bm{F^T(t-s)} - \bm{F(t-s)}) d\bm{\bm{Z}^T_s} \right)_{i}^2 \right] \leq c \sum_{1 \leq i,j \leq n} \int_0^{t} (F_{ij}^T(t-s) - F_{ij}(t-s))^2 ds.
\end{align*}
The right hand side converges to $0$ by the dominated convergence theorem together with the uniform convergence of $\bm{F^T}$ towards $\bm{F}$ (see Lemma \ref{lemma:convergence_F_T}). From Proposition \ref{prop:c_tightness_sequence} we know that $\bm{Y}=\bm{X}$ almost surely. Putting everything together, almost surely,
$$
\bm{X_t} = \int_0^{t} \bm{F(t-s)} \bm{\mu_s} ds + \int_0^{t} \bm{F(t-s)}d\bm{Z_s}.
$$
This is valid for any limit point $(\bm{X},\bm{Z})$ of $(\bm{X^T},\bm{\bm{Z}^T})$, which concludes the proof.
\end{proof}

The previous proposition gives suitable martingale properties of limit points of $\bm{Z^T}$ to apply the martingale representation theorem, which is the topic of the following proposition.

\begin{proposition}
\label{prop:first_equation_V}
Let $(\bm{X},\bm{Z})$ be a limit point of $(\bm{X^T}, \bm{\bm{Z}^T})$. There exists, up to an extension of the original probability space, an $n$-dimensional Brownian motion $\bm{B}$ and a non-negative process $\bm{V}$ such that
\begin{align*}
    \bm{X_t} &= \int_0^{t} \bm{V_s} ds \\
    \bm{Z_t} & = \int_0^{t} \textnormal{diag}(\sqrt{\bm{V_s}}) d\bm{B_s} \\
    \bm{V_t} &=  \int_0^t \bm{f(t-s)} \bm{\mu_s}ds + \int_0^{t} \bm{f(t-s)} \textnormal{diag}(\sqrt{\bm{V_s}}) d\bm{B_s}.
\end{align*}
\end{proposition}

\begin{proof}

This proof relies on the martingale representation theorem applied to $\bm{Z}$. Consider $(\bm{X},\bm{Z})$ a limit point of $(\bm{X^T}, \bm{\bm{Z}^T})$. Following the proof of Theorem 3.2 in \cite{Jaisson2016RoughProcesses}, $\bm{X}$ can be written as the integral of a process $\bm{V}$
$$
\bm{X_t} = \int_0^{t} \bm{V_s} ds,
$$
with $\bm{V}$ satisfying the equation
$$
\bm{V_t} =  \int_0^t \bm{f(t-s)} \bm{\mu}_s ds + \int_0^{t} \bm{f(t-s)} d\bm{Z_s}.
$$
Therefore, as $[\bm{Z}, \bm{Z}]_t = \textnormal{diag}(\bm{X_t}) = \textnormal{diag}(\int_0^{t} \bm{V_s} ds) $ and $\bm{Z}$ is a continuous martingale, by the martingale representation theorem (see for example Theorem 3.9 from \cite{Revuz2013ContinuousMotion}), there exists (up to an enlargement of the probability space) a multivariate Brownian motion $\bm{B}$ and a predictable square integrable process $\bm{H}$ such that
$$
\bm{Z_t} = \int_0^t \bm{H_s} d\bm{B_s}.
$$
Furthermore, note that as $\bm{V}$ is a non-negative process as $\bm{X}$ is a non-decreasing process and we have
$$
\bm{Z_t} = \int_0^t \textnormal{diag}(\sqrt{\bm{V_s}}) \textnormal{diag}(\sqrt{\bm{V_s}})^{-1} \bm{H_s} d\bm{B_s}.
$$
A simple computation shows that, since $[\bm{Z},\bm{Z}]_t = \int_0^t \bm{H_s} \trans{\bm{H_s}} ds = \bm{X_t} = \int_0^{t} \bm{V_s}ds$, the process $\bm{\Tilde{B}_t} := \int_0^t \textnormal{diag}(\sqrt{\bm{V_s}})^{-1} \bm{H_s} d\bm{B_s}$ is a Brownian motion. Finally,
$$
\bm{V_t} =  \int_0^t \bm{f(t-s)} \bm{\mu_s} ds + \int_0^{t} \bm{f(t-s)} \textnormal{diag}(\sqrt{\bm{V_s}}) d\bm{\Tilde{B_s}}.
$$
\end{proof}

A straightforward application of Lemma 4.4 and Lemma 4.5 in \cite{Jaisson2016RoughProcesses} yields the following lemma.

\begin{lemma}
\label{lemma:holder_regularity}
Consider a (weak) non-negative solution $\bm{V}$ of the stochastic Volterra equation 
$$
\bm{V_t} =  \int_0^t \bm{f(t-s)} \bm{\mu_s} ds + \int_0^{t} \bm{f(t-s)} \textnormal{diag}(\sqrt{\bm{V_s}}) d\bm{B_s},
$$
where $\bm{B}$ is a Brownian motion. Then every component of $\bm{V}$ has pathwise H\"older regularity $\alpha - 1/2 - \epsilon$ for any $\epsilon > 0$.

\end{lemma}

\subsection{Step 3: proof of Equation \eqref{eq:rough_sde_noF}}

Properties of the Mittag-Leffler function (as in \cite{ElEuch2018TheVolatility.}) enable us to rewrite the previous stochastic differential equation using power-law kernels, which is the subject of the next proposition. Let $\bm{\Theta^1} := (\bm{O_{11}} + \bm{O_{12}}(\bm{I} - \norm{\bm{C}})^{-1} \norm{\bm{B}}) \bm{K}^{-1}$, $\bm{\Theta^2} :=(\bm{O_{21}} + \bm{O_{22}} (\bm{I} - \norm{\bm{C}})^{-1} \norm{\bm{B}})\bm{K}^{-1}$ and $\bm{\Lambda} := \dfrac{\alpha}{\Gamma(1 - \alpha)} \bm{K} \bm{M}^{-1}$.

\begin{proposition}
Given an $m$-dimensional Brownian motion $\bm{B}$, a non-negative process $\bm{V}$ is solution of the following stochastic differential equation
$$
\bm{V_t} =  \int_0^t \bm{f(t-s)} \bm{\mu_s} ds + \int_0^{t} \bm{f(t-s)} \textnormal{diag}(\sqrt{\bm{V_s}}) d\bm{B_s},
$$
if and only if there exists a process $\Tilde{\bm{V}}$ of H\"older regularity $\alpha - 1/2 - \epsilon$ for any $\epsilon > 0$ such that $\bm{\Theta^1} \Tilde{\bm{V_t}} = (V^1, \cdots, V^{n_c})$ and $\bm{\Theta^2} \Tilde{\bm{V_t}} = (V^{n_c+1}, \cdots, V^{2m})$ are non-negative processes and $\Tilde{\bm{V}}$ is solution of the following stochastic Volterra equation
\begin{align*}
\Tilde{\bm{V_t}} &= \dfrac{1}{\Gamma(\alpha)}\bm{\Lambda} \int_0^t (t-s)^{\alpha-1}(\bm{O_{11}^{(-1)}}\bm{\mu^1} + \bm{O_{12}^{(-1)}}\bm{\mu^2} - \Tilde{\bm{V_s}}) ds  \\ &+\dfrac{1}{\Gamma(\alpha)}\bm{\Lambda} \int_0^t (t-s)^{\alpha-1} \bm{O_{11}^{(-1)}}\textnormal{diag}(\sqrt{\bm{\Theta^1} \Tilde{\bm{V_s}}}) d\bm{W^1_s} +\dfrac{1}{\Gamma(\alpha)}\bm{\Lambda} \int_0^t (t-s)^{\alpha-1} \bm{O_{12}^{(-1)}}\textnormal{diag}(\sqrt{\bm{\Theta^2} \Tilde{\bm{V_s}}}) d\bm{W^2_s},
\end{align*}
where $\bm{W^1} := (B^1, \cdots, B^{n_c})$ and $\bm{W^2} := (B^{n_c+1}, \cdots, B^{2m})$.
\end{proposition}

\begin{proof}

We begin by showing the first implication. Starting from Proposition \ref{prop:first_equation_V} we have
$$
\bm{V_t} =  \int_0^t \bm{f(t-s)} \bm{\mu_s} ds + \int_0^{t} \bm{f(t-s)} \textnormal{diag}(\sqrt{\bm{V_s}}) d\bm{B_s}.
$$
Developing from the definition of $\bm{f}$ in Equation \eqref{eq:definition_f_psi}, for any $t \in [0,1]$, $\bm{f}$ can be written
\begin{align*}
\bm{f(t)} &= \begin{pmatrix}
    (\bm{O_{11}} + \bm{O_{12}}(\bm{I} - \norm{\bm{C}})^{-1} \norm{\bm{B}}) \bm{K}^{-1} \bm{f^{\alpha,\Lambda}(t)}  & \bm{0} \\
    (\bm{O_{21}} + \bm{O_{22}} (\bm{I} - \norm{\bm{C}})^{-1} \norm{\bm{B}})\bm{K}^{-1} \bm{f^{\alpha,\Lambda}(t)} &  \bm{0}
    \end{pmatrix} \begin{pmatrix}
    \bm{O_{11}^{(-1)}} & \bm{O_{12}^{(-1)}} \\
    \bm{O_{21}^{(-1)}} & \bm{O_{22}^{(-1)}}
    \end{pmatrix}.
\end{align*}
Defining $\bm{V^1} := (V^1, \cdots, V^{n_c})$ and $\bm{V^2} := (V^{n_c+1}, \cdots, V^{2m})$, we have
\begin{align*}
    \bm{\bm{V^1_t}} &= 
\bm{\Theta^1} \int_0^t \bm{f^{\alpha, \bm{\Lambda}}(t-s)} \bm{O^{(-1)}_{11}}\bm{\mu^1_s} ds + \bm{\Theta^1} \int_0^t \bm{f^{\alpha, \bm{\Lambda}}(t-s)} \bm{O^{(-1)}_{12}}\bm{\mu^2_s} ds \\ & + \bm{\Theta^1} \int_0^{t} \bm{f^{\alpha, \bm{\Lambda}}}(t-s) \bm{O^{(-1)}_{11}}\textnormal{diag}(\sqrt{\bm{V^1_s}}) d\bm{W^1_s}+ \bm{\Theta^1} \int_0^{t} \bm{f^{\alpha, \bm{\Lambda}}(t-s)} \bm{O^{(-1)}_{12}}\textnormal{diag}(\sqrt{\bm{V^2_s}}) d\bm{W^2_s}.
\end{align*}
If  $\bm{\Theta^1}$ were non-singular, we could express $\bm{V^1}$ with power-law kernels thanks to the same approach as in \cite{ElEuch2018TheVolatility.}. In general we define
\begin{align*}
    \Tilde{\bm{V_t}} &:= \int_0^t \bm{f^{\alpha, \bm{\Lambda}}(t-s)} (\bm{O^{(-1)}_{11}}\bm{\mu^1_s} + \bm{O^{(-1)}_{12}}\bm{\mu^2_s})ds 
    \\ &+ \int_0^{t} \bm{f^{\alpha, \bm{\Lambda}}}(t-s) \bm{O^{(-1)}_{11}}\textnormal{diag}(\sqrt{\bm{V^1_s}}) d\bm{W^1_s} + \int_0^{t} \bm{f^{\alpha, \bm{\Lambda}}}(t-s) \bm{O^{(-1)}_{12}}\textnormal{diag}(\sqrt{\bm{V^2_s}}) d\bm{W^2_s}.
\end{align*}
From the same arguments as in Lemma \ref{lemma:holder_regularity}, H\"older regularity of $\bm{V}$ carries to $\Tilde{\bm{V}}$, and the components of $\Tilde{\bm{V}}$ are of H\"older regularity $\alpha - 1/2 - \epsilon$ for any $\epsilon > 0$, hence Lemma \ref{lemma:holder_regularity} shows $\mathcal{\bm{K}} := I^{1 - \alpha} \Tilde{\bm{V}}$ is well-defined, where $I^{1-\alpha}$ is the fractional integration operator of order $1-\alpha$ (see Definition \ref{def:fractional_differentiation_integration} in Appendix \ref{sec:fractional_operators}). Note that for any $t$ in $[0,1]$, using Lemma \ref{lemma:integration_operator_applied_f_alpha} of Appendix \ref{sec:fractional_operators}, we have
\begin{align*}
    \mathcal{\bm{K}}_t &= \int_0^t \bm{\Lambda} (\bm{I} - \bm{F^{\alpha, \bm{\Lambda}}(t-s)}) (\bm{O^{(-1)}_{11}} \bm{\mu^1_s} + \bm{O^{(-1)}_{12}} \bm{\mu^2_s}) ds \\
    & + \int_0^t \bm{\Lambda} (\bm{I} - \bm{F^{\alpha, \bm{\Lambda}}}(t-s)) \bm{O^{(-1)}_{11}}\textnormal{diag}(\sqrt{\bm{V^1_s}}) d\bm{W^1_s} + \int_0^t \bm{\Lambda} (\bm{I} - \bm{F^{\alpha, \bm{\Lambda}}}(t-s)) \bm{O^{(-1)}_{12}}\textnormal{diag}(\sqrt{\bm{V^2_s}}) d\bm{W^2_s}  \\
    & = \bm{\Lambda} \int_0^t (\bm{O^{(-1)}_{11}}\bm{\mu^1_s} +\bm{O^{(-1)}_{12}}\bm{\mu^2_s}) ds + \int_0^t \bm{\Lambda} \bm{O_{11}}\textnormal{diag}(\sqrt{\bm{V^1_s}}) d\bm{W^1_s} + \int_0^t \bm{\Lambda} \bm{O^{(-1)}_{12}}\textnormal{diag}(\sqrt{\bm{V^2_s}}) d\bm{W^2_s} \\
    &- \bm{\Lambda} \int_0^t \left[ \bm{F^{\alpha, \bm{\Lambda}}(t-s)}\bm{O^{(-1)}_{11}} \bm{\mu^1_s} + \int_0^s \bm{f^{\alpha, \bm{\Lambda}}(s-u)} \bm{O^{(-1)}_{11}}\textnormal{diag}(\sqrt{\bm{V^1_u}}) d\bm{W^1_u}\right] ds
    \\ &- \bm{\Lambda} \int_0^t \left[ \bm{F^{\alpha, \bm{\Lambda}}(t-s)}\bm{O^{(-1)}_{12}} \bm{\mu^2_s} + \int_0^s \bm{f^{\alpha, \bm{\Lambda}}(s-u)} \bm{O^{(-1)}_{12}}\textnormal{diag}(\sqrt{\bm{V^2_u}})d\bm{W^2_u} \right] ds.
\end{align*}
The last two terms can be rewritten using the definition of $\Tilde{\bm{V}}$, so that
\begin{align*}
    \mathcal{\bm{K}}_t &= \bm{\Lambda} \int_0^t (\bm{O^{(-1)}_{11}}\bm{\mu^1_s} + \bm{O^{(-1)}_{12}}\bm{\mu^2_s} - \Tilde{\bm{V_s}}) ds + \bm{\Lambda} \int_0^t  \bm{O^{(-1)}_{11}}\textnormal{diag}(\sqrt{\bm{\Theta^1} \Tilde{\bm{V_s}}})d\bm{W^1_s} + \bm{\Lambda} \int_0^t  \bm{O^{(-1)}_{12}}\textnormal{diag}(\sqrt{\bm{\Theta^2} \Tilde{\bm{V_s}}}) d\bm{W^2_s}.
\end{align*}
Thanks to the H\"older regularity of $\bm{\Tilde{V}}$, we can now apply the fractional differentiation operator of order $1-\alpha$ (see Definition \ref{def:fractional_differentiation_integration} in Appendix \ref{sec:fractional_operators}) together with the stochastic Fubini Theorem to deduce
\begin{align*}
\Tilde{\bm{V_t}} &= \dfrac{1}{\Gamma(\alpha)}\bm{\Lambda} \int_0^t (t-s)^{\alpha-1}(\bm{O^{(-1)}_{11}}\bm{\mu^1_s} + \bm{O^{(-1)}_{12}}\bm{\mu^2_s} - \Tilde{\bm{V_s}}) ds  \\ &+\dfrac{1}{\Gamma(\alpha)}\bm{\Lambda} \int_0^t (t-s)^{\alpha-1} \bm{O^{(-1)}_{11}}\textnormal{diag}(\sqrt{\bm{\Theta^1} \Tilde{\bm{V_s}}}) d\bm{W^1_s} +\dfrac{1}{\Gamma(\alpha)}\bm{\Lambda} \int_0^t (t-s)^{\alpha-1} \bm{O^{(-1)}_{12}}\textnormal{diag}(\sqrt{\bm{\Theta^2} \Tilde{\bm{V_s}}}) d\bm{W^2_s}. 
\end{align*}
This concludes the proof of the first implication. We now show the second implication. Suppose there exists $\Tilde{\bm{V}}$ of H\"older regularity $\alpha - 1/2 - \epsilon$ for any $\epsilon > 0$ such that $\bm{\Theta^1} \Tilde{\bm{V}}$ and $\bm{\Theta^2} \Tilde{\bm{V}}$ are positive,  solution of the following stochastic Volterra equation:
\begin{align*}
\Tilde{\bm{V_t}} &= \dfrac{1}{\Gamma(\alpha)}\bm{\Lambda} \int_0^t (t-s)^{\alpha-1}(\bm{O^{(-1)}_{11}}\bm{\mu^1_s} + \bm{O^{(-1)}_{12}}\bm{\mu^2_s} - \Tilde{\bm{V_s}}) ds  \\ &+\dfrac{1}{\Gamma(\alpha)}\bm{\Lambda} \int_0^t (t-s)^{\alpha-1} \bm{O^{(-1)}_{11}}\textnormal{diag}(\sqrt{\bm{\Theta^1} \Tilde{\bm{V_s}}}) d\bm{W^1_s} +\dfrac{1}{\Gamma(\alpha)}\bm{\Lambda} \int_0^t (t-s)^{\alpha-1} \bm{O^{(-1)}_{12}}\textnormal{diag}(\sqrt{\bm{\Theta^2} \Tilde{\bm{V_s}}}) d\bm{W^2_s}.
\end{align*}
Let us write for this proof $\bm{\theta} := \bm{\Lambda} \bm{O^{(-1)}_{11}}\bm{\mu^1} + \bm{\Lambda} \bm{O^{(-1)}_{12}}\bm{\mu^2}, \bm{\Lambda}_1 := \bm{\Lambda} \bm{O^{(-1)}_{11}}, \bm{\Lambda}_2 := \bm{\Lambda} \bm{O^{(-1)}_{12}}$ so that, for any $t$ in $[0,1]$,
\begin{equation*}
\Tilde{\bm{V_t}} = \dfrac{1}{\Gamma(\alpha)} \int_0^t (t-s)^{\alpha-1}(\bm{\theta}_s - \bm{\Lambda} \Tilde{\bm{V_s}}) ds \\
+\dfrac{1}{\Gamma(\alpha)} \int_0^t (t-s)^{\alpha-1}\bm{\Lambda_1} \textnormal{diag}(\sqrt{\bm{\Theta^1} \Tilde{\bm{V_s}}}) d\bm{W^1_s} +\dfrac{1}{\Gamma(\alpha)}\int_0^t (t-s)^{\alpha-1} \bm{\Lambda_2}\textnormal{diag}(\sqrt{\bm{\Theta^2} \Tilde{\bm{V_s}}}) d\bm{W^2_s}.
\end{equation*}
Remark that the above can be written
\begin{equation*}
    \Tilde{\bm{V_t}} = I^{\alpha}(\bm{\theta} - \bm{\Lambda} \Tilde{\bm{V}})_t + I^{\alpha}_{\bm{B}^1}(\bm{\Lambda_1}\textnormal{diag}(\sqrt{\bm{\Theta^1} \Tilde{\bm{V}}}))_t + I^{\alpha}_{\bm{B}^2}(\bm{\Lambda_2}\textnormal{diag}(\sqrt{\bm{\Theta^2} \Tilde{\bm{V}}}))_t,
\end{equation*}
where $I^{\alpha}_{\bm{B}}$ is the fractional integration operator with respect to $\bm{B}$ (see Definition \ref{def:fractional_differentiation_brownian} in Appendix \ref{sec:fractional_operators}). Iterating the application of $I^{\alpha}$ we find that, for any $N \geq 1$, $\Tilde{\bm{V}}$ satisfies
\begin{align*}
\Tilde{\bm{V}} &= \sum_{1 \leq k \leq N} \bm{\Lambda}^{k-1} (-1)^{k-1} I^{(k-1) \alpha} [I^{\alpha} \bm{\theta} + I^{\alpha}_{\bm{B}^1}(\bm{\Lambda_1}\textnormal{diag}(\sqrt{\bm{\Theta^1} \Tilde{\bm{V}}})) + I^{\alpha}_{\bm{B}^2}(\bm{\Lambda_2}\textnormal{diag}(\sqrt{\bm{\Theta^2} \Tilde{\bm{V}}}))] + \bm{\Lambda}^{N} (-1)^{N} I^{(N+1) \alpha}\Tilde{\bm{V}}.    
\end{align*}
Now, note that $\bm{\theta}$, $\textnormal{diag}(\sqrt{\bm{\Theta^1} \Tilde{\bm{V}}})$, $\textnormal{diag}(\sqrt{\bm{\Theta^2} \Tilde{\bm{V}}})$ and $\bm{\Tilde{V}}$ are square-integrable processes and Lemma \ref{lemma:convergence_mittagleffler_brownian} in Appendix \ref{sec:fractional_operators} shows that the sum converges almost surely to the series while $\bm{\Lambda}^{N} (-1)^{N} I^{(N+1) \alpha}\Tilde{\bm{V}}$ converges almost surely to zero as $N$ tends to infinity. Thus we have
\begin{align*}
\Tilde{\bm{V}} &= \sum_{k \geq 0} \bm{\Lambda}^{k} (-1)^{k} I^{k \alpha} [I^{\alpha} \bm{\theta} + I^{\alpha}_{\bm{B}^1}(\bm{\Lambda_1}\textnormal{diag}(\sqrt{\bm{\Theta^1} \Tilde{\bm{V}}})) + I^{\alpha}_{\bm{B}^2}(\bm{\Lambda_2}\textnormal{diag}(\sqrt{\bm{\Theta^2} \Tilde{\bm{V}}}))] \\
&= \sum_{k \geq 0} \bm{\Lambda}^{k} (-1)^{k} I^{k \alpha} I^{\alpha} \bm{\theta} + \sum_{k \geq 0} \bm{\Lambda}^{k} (-1)^{k} I^{k \alpha} I^{\alpha}_{\bm{B}^1}(\bm{\Lambda_1}\textnormal{diag}(\sqrt{\bm{\Theta^1} \Tilde{\bm{V}}})) + I^{\alpha}_{\bm{B}^2}(\bm{\Lambda_2}\textnormal{diag}(\sqrt{\bm{\Theta^2} \Tilde{\bm{V}}}))] \\
&= \bm{\Lambda}^{-1} \sum_{k \geq 0} \bm{\Lambda}^{k+1} (-1)^{k} I^{(k+1) \alpha} \bm{\theta} + \sum_{k \geq 0} \bm{\Lambda}^{k} (-1)^{k} I^{k \alpha} I^{\alpha}_{\bm{B}^1}(\bm{\Lambda_1}\textnormal{diag}(\sqrt{\bm{\Theta^1} \Tilde{\bm{V}}})) + I^{\alpha}_{\bm{B}^2}(\bm{\Lambda_2}\textnormal{diag}(\sqrt{\bm{\Theta^2}
\Tilde{\bm{V}}}))].
\end{align*}
Lemmas \ref{lemma:series_mittagleffler} and \ref{lemma:series_convolutionmittagleffler} shown in Appendix \ref{sec:fractional_operators} enable us to rewrite the above using the matrix Mittag-Leffler function. This yields, for any $t$ in $[0,1]$ and almost surely,
\begin{align*}
    \Tilde{\bm{V_t}} &= \bm{\Lambda}^{-1} \int_0^t \bm{f^{\alpha, \Lambda}(t-s)} \bm{\theta_s} ds + \bm{\Lambda}^{-1} \int_0^{t} \bm{f^{\alpha, \Lambda}(t-s)} \bm{\Lambda_1}\textnormal{diag}(\sqrt{\bm{\Theta^1} \Tilde{\bm{V_s}}}) d\bm{W^1_s} + \bm{\Lambda}^{-1} \int_0^{t} \bm{f^{\alpha, \Lambda}(t-s)} \bm{\Lambda_2}\textnormal{diag}(\sqrt{\bm{\Theta^2} \Tilde{\bm{V_s}}}) d\bm{W^2_s} .
\end{align*}
Replacing $\bm{\theta}, \bm{\Lambda}_1, \bm{\Lambda}_2$ by their expressions, almost surely and for any $t$ in $[0,1]$,
\begin{align*}
    \Tilde{\bm{V}}_t &= \int_0^t \bm{f^{\alpha, \Lambda}(t-s)} (\bm{O^{(-1)}_{11}}\bm{\mu^1_s} + \bm{\Lambda} \bm{O^{(-1)}_{12}}\bm{\mu^2_s})ds \\
    &+ \int_0^{t} \bm{f^{\alpha, \Lambda}(t-s)} \bm{O^{(-1)}_{11}}\textnormal{diag}(\sqrt{\bm{\Theta^1} \Tilde{\bm{V_s}}}) d\bm{B_s}^1 + \int_0^{t} \bm{f^{\alpha, \Lambda}(t-s)} \bm{O^{(-1)}_{12}}\textnormal{diag}(\sqrt{\bm{\Theta^2} \Tilde{\bm{V_s}}}) d\bm{B_s}^2.
\end{align*}
This concludes the second implication and the proof.
\end{proof}

\subsection{Step 4: Equation satisfied by the limiting price process}
\label{sec:proof2}

The previous results on the convergence of the intensity process enable us to now turn to the question of the limiting price dynamics. Recall that the sequence of rescaled price processes $\bm{P^T}$ is defined as
$$
\bm{P^T} := \trans{\bm{Q}} \bm{\bm{X}^T},
$$
where $\bm{Q} = \begin{pmatrix}
\bm{e_1} - \bm{e_2} \mid & \cdots  &\mid \bm{e_{2m-1}} - \bm{e_{2m}}
\end{pmatrix}.$ We have the following result.

\begin{proposition}
Let $(\bm{X},\bm{Z})$ be a limit point of $(\bm{X}^T, \bm{Z}^T)$ and $\bm{P} = \trans{\bm{Q}} \bm{X}$. Then
$$
\bm{P_t} = (\bm{I} + \bm{\Delta}) \trans{\bm{Q}} (\bm{Z_t} + \int_0^t \bm{\mu_s} ds).
$$
where $\bm{\Delta} = (\norm{\delta^T_{ij}})_{1 \leq i,j \leq m}$.
\end{proposition}

\begin{proof}
Let $(\bm{X}, \bm{Z})$ be a limit point from $(\bm{X}^T, \bm{Z}^T)$. For any $1 \leq i \leq m$ we can compute the difference between upward and downard jumps on Asset $i$
\begin{align*}
    \bm{v_i} \cdot \bm{N^T}_t = \bm{v_i} \cdot \bm{M^T_t} + \bm{v_i} \cdot \int_0^t \bm{\lambda_s} ds,
\end{align*}
with the following expression for the integrated intensity:
\begin{align*}
    \int_0^{tT} \bm{\lambda^T}_s ds & = T \int_0^{t} \bm{\mu^T_{sT}} ds + T \int_0^t \int_0^{T(t-s)} \bm{\psi^T(u)} du \bm{\mu^T_{Ts}} ds + \bm{\normOne{\psi^T}} \bm{M^T_{tT}} - \int_0^{tT} \int_{tT-s}^{\infty} \bm{\psi^T(u)} du d\bm{M^T_s}.
\end{align*}
Thus the microscopic price for the Asset $i$ satisfies
\begin{align*}
    T^{-\alpha} \bm{v_i} \cdot \bm{N^T}_{tT} &= T^{1-\alpha} \int_0^{t} \bm{v_i} \cdot \bm{\mu^T_{sT}} ds + T^{1-\alpha} \trans{\bm{\normOne{\psi^T}}} \bm{v_i} \cdot \int_0^t \bm{\mu^T_{Ts}} ds + \bm{v_i} \cdot \bm{Z^T_t} + \trans{\bm{\normOne{\psi^T}}} \bm{v_i} \cdot \bm{Z^T_{t}} \\
    & -  T^{-\alpha} \int_0^t \int_{T(t-s)}^{\infty} \trans{\bm{\psi^T(u)}}\bm{v_i} \cdot \bm{\mu^T_{Ts}} du ds - T^{-\alpha} \int_0^{tT}  \int_{tT-s}^{\infty} \bm{\psi^T(u)} du d\bm{M^T_s} \\
    &= \sum_{1 \leq k \leq m} (\mathbb{1}_{ik} + \norm{\delta^T_{ik}}), \bm{v_k} \cdot \bm{Z^T_t} + \sum_{1 \leq k \leq m} (\mathbb{1}_{ik} + \norm{\delta^T_{ik}}) T^{1-\alpha} \int_0^{t} \bm{v_k} \cdot \bm{\mu^T_{sT}} ds \\
    & -  \int_0^{t} \int_{tT-s}^{\infty} \trans{\bm{\psi^T(u)}}\bm{v_i} du \cdot  d\bm{Z^T_s} -  T^{-\alpha} \int_0^t \int_{T(t-s)}^{\infty} \trans{\bm{\psi^T(u)}}\bm{v_i} \cdot \bm{\mu^T_{Ts}} du ds.
\end{align*}
It is straightforward to show that the last two terms converge to zero and
thus, any limit point $\bm{P}$ of $\bm{P^T} = \trans{\bm{Q}} \bm{X}^T$ is such that
$$
\bm{P_t} = (\bm{I} + \bm{\Delta}) \trans{\bm{Q}} (\bm{Z_t} + \int_0^t \bm{\mu_s} ds).
$$
\end{proof}

Replacing $\bm{Z}$ by the expression obtained in Proposition \ref{prop:first_equation_V} concludes the proof of Theorem \ref{thm:ConvergenceInt} since
$$
\bm{P_t} = (\bm{I} + \bm{\Delta}) \trans{\bm{Q}} \big( \int_0^{t} \textnormal{diag}(\sqrt{\bm{V_s}}) d\bm{B_s} + \int_0^t \bm{\mu_s} ds \big).
$$

\newpage
\appendix
\section{Technical appendix}
\label{sec:appendix}

\subsection{Independence of Equation \eqref{eq:rough_sde_noF} from chosen basis}

We consider two representations which satisfy Assumption \ref{ass:StructureHawkes}. Let $\bm{P}, \Tilde{\bm{P}}$ be invertible matrices, $0 \leq n_c, n_{c^{'}} \leq n$ and $\bm{A^T}  \in \matfunction{n_c}{\mathbb{R}}$, $\bm{C^T} \in \matfunction{n-n_c}{\mathbb{R}}$, $\bm{B^T} \in \matfunction{n-n_c,n_c}{\mathbb{R}}$ and $\bm{\Tilde{A}}^T  \in \matfunction{n_{c^{'}}}{\mathbb{R}}$, $\bm{\Tilde{C}}^T \in \matfunction{n-n_{c^{'}}}{\mathbb{R}}$, $\bm{\Tilde{B}}^T \in \matfunction{n-n_{c^{'}},n_{c^{'}}}{\mathbb{R}}$ such that
\begin{align*}
    \bm{\phi^T} &= \bm{P} \begin{pmatrix}
    \bm{A^T} & \bm{0} \\
    \bm{B^T} & \bm{C^T}
    \end{pmatrix} \bm{P}^{-1} = \Tilde{\bm{P}} \begin{pmatrix}
    \bm{\Tilde{A}}^T & \bm{0} \\
    \bm{\Tilde{B}}^T & \bm{\Tilde{C}}^T
    \end{pmatrix} \Tilde{\bm{P}}^{-1}.
\end{align*}
We write $\bm{A}$ for the limit of $\bm{A^T}$ (and similarly for $B^T,C^T$, etc.). First, remark that we must have $n_c = n_{c^{'}}$. Indeed, since $\spectralRadius{\norm{\bm{C}}} <1$ and $\spectralRadius{\norm{\bm{\Tilde{C}}}} <1$, $1$ is neither an eigenvalue of $\norm{\bm{C}}$ nor of $\norm{\bm{\Tilde{C}}}$. Yet, since $\bm{A} = \bm{I}$ and $\bm{\Tilde{A}} = \bm{I}$, 1 is an eigenvalue of $\bm{\phi}$ with multiplicity $n_{c}$ and $n_{c^{'}}$. Therefore $n_{c} = n_{c^{'}}$.
\\ \\
We have, writing $\bm{L}=\bm{P}^{-1} \Tilde{\bm{P}}$,
\begin{align*}
    \begin{pmatrix}
    \bm{A} & \bm{0} \\
    \bm{B} & \bm{C}
    \end{pmatrix} &= \bm{L} \begin{pmatrix}
    \bm{\Tilde{A}} & \bm{0} \\
    \bm{\Tilde{B}} & \bm{\Tilde{C}}
    \end{pmatrix} \bm{L}^{-1}.
\end{align*}
Since $\bm{A} = \bm{\Tilde{A}} = \bm{I}$ because of Equation \eqref{eq:convergence_rates_AT}, developing and using the assumption that $\bm{I}-\bm{C}$ is invertible, we get
\begin{align*}
    \bm{L_{12}} &= \bm{0} \\
     (\bm{I}-\bm{C}) \bm{L_{21}} &= \bm{B} \bm{L_{11}} - \bm{L_{22}} \bm{\Tilde{B}} \\
     \bm{C} \bm{L_{22}} &= \bm{L_{22}} \bm{\Tilde{C}}.
\end{align*}
Since $\bm{L} \bm{L}^{-1} = \bm{I}$, $\bm{L_{11}} = \bm{I}$, $\bm{L_{22}} = \bm{I}$, $\bm{L_{21}} = - \bm{L}^{(-1)}_{21}$, we deduce
\begin{align*}
    \bm{L_{11}} = \bm{I}, \quad \bm{L_{22}} = \bm{I}, \quad   \bm{L_{12}} = \bm{0}, \quad
     (\bm{I}-\bm{C}) \bm{L_{21}} &= \bm{B} - \bm{\Tilde{B}}, \quad \bm{C} = \bm{\Tilde{C}}.
\end{align*}
As $\bm{L} = \bm{P}^{-1} \Tilde{\bm{P}}$, we have
\begin{align*}
     \bm{P}^{-1} &=  \begin{pmatrix}
    \bm{I} & \bm{0} \\
    (\bm{I}-\bm{C})^{-1} (\bm{B}-\bm{\Tilde{B}}) & \bm{I}
    \end{pmatrix} \Tilde{\bm{P}}^{-1} = \begin{pmatrix}
    \Tilde{\bm{P}}^{(-1)}_{11} & \Tilde{\bm{P}}^{(-1)}_{12} \\
     (\bm{I}-\bm{C})^{-1} (\bm{B}-\bm{\Tilde{B}})\Tilde{\bm{P}}^{(-1)}_{11} + \Tilde{\bm{P}}^{(-1)}_{21} &  (\bm{I}-\bm{C})^{-1} (\bm{B}-\bm{\Tilde{B}})\Tilde{\bm{P}}^{(-1)}_{12} + \Tilde{\bm{P}}^{(-1)}_{22}
    \end{pmatrix}.
\end{align*}
Developing $\Tilde{\bm{P}} = \bm{P} \bm{L}$ together with the above, we find
\begin{align*}
    &\Tilde{\bm{P}}^{(-1)}_{11} = \bm{P}^{(-1)}_{11}, 
    \Tilde{\bm{P}}^{(-1)}_{12} = \bm{P}^{(-1)}_{12},
    \Tilde{\bm{P}}_{12} = \bm{P_{12}},
    \Tilde{\bm{P}}_{22} = \bm{P_{22}} \\
    &\Tilde{\bm{P}}_{11} = \bm{P_{11}} +  \bm{P_{12}} (\bm{I}-\bm{C})^{-1} (\bm{B}-\bm{\Tilde{B}}) \\
    &\Tilde{\bm{P}}_{21} = \bm{P_{21}} +  \bm{P_{22}}(\bm{I}-\bm{C})^{-1} (\bm{B}-\bm{\Tilde{B}}).
\end{align*}
Thus
\begin{align*}
    & \Tilde{\bm{P}}^{(-1)}_{11} = \bm{P}^{(-1)}_{11}, \Tilde{\bm{P}}^{(-1)}_{12} = \bm{P}^{(-1)}_{12} \\
    & \Tilde{\bm{P}}_{11} + \Tilde{\bm{P}}_{12}(\bm{I}-\bm{C})^{-1}\bm{\Tilde{B}} = \bm{P_{11}} +  \bm{P_{12}} (\bm{I}-\bm{C})^{-1} \bm{B} \\
    & \Tilde{\bm{P}}_{21} + \Tilde{\bm{P}}_{22}(\bm{I}-\bm{C})^{-1} \bm{\Tilde{B}} = \bm{P_{21}} +  \bm{P_{22}}(\bm{I}-\bm{C})^{-1} \bm{B}.
\end{align*}

Therefore regardless of the chosen basis, Equation \eqref{eq:rough_sde_noF} is the same, which concludes the proof.

\subsection{Fractional operators}
\label{sec:fractional_operators}

This section is a brief reminder on fractional operators which are used in proofs. We also introduce the matrix extended Mittag-Leffler function.

\begin{definition}[Fractional differentiation and integration operators]
\label{def:fractional_differentiation_integration}
For $\alpha \in (0,1)$, the fractional differentiation (resp. integration) operator denoted by $D^{\alpha}$ is defined as
\begin{align*}
    D^{\alpha}f(t) &:= \dfrac{1}{\Gamma(1 - \alpha)}\dfrac{d}{dt} \int_0^t (t-s)^{-\alpha} f(s) ds,
\end{align*}
where $f$ is a measurable, H\"older continuous function of order strictly greater than $\alpha$. The fractional integration operator denoted by $I^{\alpha}$ is defined as
$$
I^{\alpha}f(t) := \dfrac{1}{\Gamma(\alpha)}\int_0^t (t-s)^{\alpha-1} f(s) ds.
$$
where $f$ is a measurable function.
\end{definition}

It will be convenient for us to define fractional integration with respect to a Brownian motion.

\begin{definition}[Fractional integration operator with respect to a Brownian motion]
\label{def:fractional_differentiation_brownian}
Given a Brownian motion $B$ and $\alpha \in (1/2,1)$, the fractional integration operator with respect to $B$, denoted by $I^{\alpha}_{B}$, is defined as
$$
I^{\alpha}_{B}f(t) = \dfrac{1}{\Gamma(\alpha)}\int_0^t (t-s)^{1-\alpha} f(s) dB_s.
$$
for $f$ a measurable, square integrable stochastic process.
\end{definition}

\begin{remark}
The fractional integration of a matrix-valued stochastic process $\bm{f}$ with respect to a multivariate Brownian motion $\bm{B}$ is:

$$
I^{\alpha}_{\bm{B}}\bm{f(t)} = \dfrac{1}{\Gamma(\alpha)}\int_0^t (t-s)^{1-\alpha} \bm{f(s)} d\bm{B_s}.
$$
\end{remark}

We now extend the Mittag-Leffler function to matrices (for a theory of matrix-valued functions, see for example \cite{Higham2008FunctionsComputation}). We have the following definition.

\begin{definition}[Matrix-extended Mittag-Leffler function]
\label{def:matrix_extended_mittagleffler}
Let $\alpha, \beta \in \mathbb{C}$ such that $\textnormal{Re}(\alpha), \textnormal{Re}(\beta) > 0$, $\bm{\Lambda} \in \mathcal{M}_{n}(\mathbb{R})$. Then the matrix Mittag-Leffler function is defined as
\begin{align*}
    \bm{E_{\alpha, \beta}(\Lambda)} &:= \sum_{n \geq 0} \dfrac{\bm{\Lambda}^n}{\Gamma(\alpha n + \beta)}.
\end{align*}
\end{definition}

We also extend the Mittag-Leffler density function for matrices.

\begin{definition}[Mittag-Leffler density for matrices]
\label{def:mittag_leffler_matrices}
 Let $\alpha \in \mathbb{C}$ such that $\textnormal{Re}(\alpha) > 0$, $\bm{\Lambda} \in \mathcal{M}_{n}(\mathbb{R})$. Then, the matrix Mittag-Leffler density function $\bm{f^{\alpha, \Lambda}}$ is defined as

$$
    \bm{f^{\alpha, \Lambda}(t)} := \bm{\Lambda} t^{\alpha - 1} \bm{E_{\alpha, \alpha}(-\bm{\Lambda} t^{\alpha})}
$$

We write $\bm{F^{\alpha, \Lambda}}$ for the cumulative matrix Mittag-Leffler density function

$$
    \bm{F^{\alpha, \Lambda}(t)} := \int_0^t  \bm{f^{\alpha, \Lambda}(s)}ds
$$

\end{definition}

Using Definition \ref{def:matrix_extended_mittagleffler}, it is easy to show the following lemma.
\begin{lemma}
\label{lemma:integration_operator_applied_f_alpha}
Let $\alpha\in \mathbb{C}$ such that $\textnormal{Re}(\alpha)> 0$, $\bm{\Lambda} \in \mathcal{M}_{n}(\mathbb{R})$. Then,
$$
    I^{1-\alpha} \bm{f^{\alpha, \bm{\Lambda}}} = \bm{\Lambda} (\bm{I} - \bm{F^{\alpha, \bm{\Lambda}})}.
$$
Furthermore, if $\alpha \in (1/2,1)$
$$
\widehat{\bm{f^{\alpha, \bm{\Lambda}}}}\bm{(z)} = \bm{\Lambda} (\bm{I} z^{\alpha} + \bm{\Lambda})^{-1}.
$$
\end{lemma}

We need another important property relating Mittag-Leffler functions with fractional integration operators.

\begin{lemma}
\label{lemma:series_mittagleffler}
Let $\alpha > 0$ and $\bm{\Lambda} \in \mathcal{M}_{m}(\mathbb{R})$. Then
\begin{align*}
    I^1 \bm{f^{\alpha, \bm{\Lambda}}} &= \sum_{ n \geq 1} (-1)^{n-1} \bm{\Lambda}^{n} I^{n \alpha}(1)
\end{align*}
\end{lemma}

\begin{proof}
Using Lemma \ref{lemma:integration_operator_applied_f_alpha} and repeated applications of $I^{\alpha}$, for all $N \geq 1$ we have
\begin{align*}
    I \bm{f^{\alpha, \bm{\Lambda}}} &= \sum_{1 \leq n \leq N} (-1)^{n-1} \bm{\Lambda}^{n} I^{n \alpha}(1) + (-1)^{N-1}\bm{\Lambda}^{N} I^{N \alpha} I \bm{f^{\alpha, \bm{\Lambda}}}.
\end{align*}
Therefore, if we show that 
$$
(-1)^{N-1}\bm{\Lambda}^{N} I^{N \alpha} I \bm{f^{\alpha, \bm{\Lambda}}} \underset{N \to \infty}{\to} 0,
$$
the result will follow. To prove this we make use of the series expansion of $I^{N \alpha} \bm{f^{\alpha, \Lambda}}$ to deduce bounds which will converge to zero. Writing $C$ a constant independent of $t$ and $N$ which may change from line to line, $N_{\alpha} = \lfloor \dfrac{1}{\alpha} \rfloor$ and $\normOp{\cdot}$ for the operator norm, we have
\begin{align*}
    \normOp{\bm{\Lambda}^{N} \bm{f^{\alpha, \Lambda}(t)}} &=  \normOp{\bm{\Lambda}^{N+1} \sum_{n \geq 0} (-1)^{n} \dfrac{t^{(n+1)\alpha - 1}}{\Gamma((n+1)\alpha)}} \\
    & \leq \normOp{\bm{\Lambda}^{N+1} \sum_{0 \leq n \leq N_{\alpha}} (-1)^n \dfrac{t^{(n+1)\alpha - 1}}{\Gamma((n+1)\alpha)} + \bm{\Lambda}^{N+1} C}.
\end{align*}
Therefore, when applying the fractional integration operator of order $N\alpha$ we have, writing $g_{n}: t \mapsto t^{(n+1)\alpha-1}$
\begin{align*}
    I^{N \alpha} \normOp{\bm{\Lambda}^{N} \bm{f^{\alpha, \Lambda}(t)}} &\leq \normOp{\bm{\Lambda}^{N+1}I^{N \alpha} (\sum_{0 \leq n \leq N_{\alpha}} (-1)^n \dfrac{g_n}{\Gamma((n+1)\alpha)}) + \bm{\Lambda}^{N+1}I^{N \alpha}(C)} \\
    & \leq \sum_{0 \leq n \leq N_{\alpha}} \dfrac{1}{\Gamma((n+1)\alpha)} \normOp{\bm{\Lambda}^{N+1}I^{N \alpha} (g_{n})} + \normOp{\bm{\Lambda}^{N+1}I^{N \alpha}(C)}.
\end{align*}
An explicit computation of $I^{N\alpha}(g_n)$ shows the convergence to zero of the right hand side as $N$ tends to infinity, which concludes the proof.
\end{proof}

Finally, we need to combine fractional integration $I^{\alpha}$ with $I^{\alpha}_{B}$. We have the following lemma.

\begin{lemma}
\label{lemma:composition_brownian}
Let $m \geq 1$, $\bm{B}$ an $m$-dimensional Brownian motion, $\bm{X}$ a $m \times m$ matrix valued adapted square-integrable stochastic process and $\alpha, \beta > 0$. Then we have:
$$
I^{\alpha} I^{\beta}_{\bm{B}}(\bm{X}) = I^{\alpha + \beta}_{\bm{B}}(\bm{X}).
$$
\end{lemma}

\begin{proof}
The proof is a straightforward application of the definition of the operators together with the stochastic Fubini theorem.
\end{proof}
The next lemma is useful to transform stochastic convolutions of stochastic processes with the Mittag-Leffler density function into series of repeated applications of $I^{\alpha}_{\bm{B}}$.
\begin{lemma}
\label{lemma:series_convolutionmittagleffler}
Let $m \geq 1$, $\bm{B}$ an $m$-dimensional Brownian motion, $\bm{X}$ a $m \times m$ matrix valued adapted and square-integrable stochastic process, $\alpha > 0$ and $\bm{\Lambda} \in \mathcal{M}_{m}(\mathbb{R})$. Then, for all $t \geq 0$ and almost surely
\begin{equation*}
    \int_0^t \bm{f^{\alpha, \Lambda}(t-s)} \bm{X_s} d\bm{B_s} = \sum_{n \geq 1} (-1)^{n-1} \bm{\Lambda}^{n} I^{n \alpha}_{\bm{B}}(\bm{X}),
\end{equation*}
where the series converges almost surely.
\end{lemma}

\begin{proof}
Using Lemma \ref{lemma:series_mittagleffler}, we can write the integral using a series of fractional integration operators and apply the stochastic Fubini theorem (as $\bm{X}$ is square-integrable) to obtain
\begin{align*}
    \int_0^t \bm{f^{\alpha, \Lambda}(t-s)} \bm{X_s} d\bm{B_s} &= \int_0^t \sum_{n \geq 1} (-1)^{n-1} \bm{\Lambda}^{n} I^{n\alpha-1}(1)_{t-s} \bm{X_s} d\bm{B_s} \\
    &= \sum_{n \geq 1} \int_0^t (-1)^{n-1} \bm{\Lambda}^{n} I^{n\alpha-1}(1)_{t-s} \bm{X_s} d\bm{B_s} \\
    &= \sum_{n \geq 1} (-1)^{n-1} \bm{\Lambda}^{n} \int_0^t I^{n\alpha-1}(1)_{t-s} \bm{X_s} d\bm{B_s} \\
    &= \sum_{n \geq 1} \dfrac{(-1)^{n-1}}{\Gamma(n\alpha-1)} \bm{\Lambda}^{n} \int_0^t \int_0^{t-s} (t-s-\tau)^{n\alpha - 2} d \tau \bm{X_s} d\bm{B_s}.
\end{align*}
 After a change of variables and using the stochastic Fubini theorem (see for example \cite{Veraar2012TheRevisited}), we deduce the simpler expression
 \begin{align*}
    \int_0^t \bm{f^{\alpha, \Lambda}(t-s)} \bm{X_s} d\bm{B_s} &= \sum_{n \geq 1} \dfrac{(-1)^{n-1}}{\Gamma(n\alpha-1)} \bm{\Lambda}^{n} \int_0^t (t-\tau)^{n\alpha - 2} \int_0^{\tau} \bm{X_s} d\bm{B_s}  d\tau.
\end{align*}
Integrating by parts, we finally obtain the result:
\begin{align*}
    \int_0^t \bm{f^{\alpha, \Lambda}(t-s)} \bm{X_s} d\bm{B_s} &= \sum_{n \geq 1} \dfrac{(-1)^{n-1}}{\Gamma(n\alpha-1) (n\alpha-1)} \bm{\Lambda}^{n}  \int_0^t (t-\tau)^{n\alpha-1} \bm{X_{\tau}} d\bm{B_{\tau}}, \\
    &= \sum_{n \geq 1} \dfrac{(-1)^{n-1}}{\Gamma(n\alpha)} \bm{\Lambda}^{n}  \int_0^t (t-\tau)^{n\alpha-1} \bm{X_{\tau}} d\bm{B_{\tau}}, \\
    &= \sum_{n \geq 1} (-1)^{n-1} \bm{\Lambda}^{n} I^{n\alpha}_{\bm{B}}(\bm{X}).
\end{align*}
\end{proof}

The last lemma gives convergence for terms of a series of repeated iterations of $I^{\alpha}$.
\begin{lemma}
\label{lemma:convergence_mittagleffler_brownian}
Let $\alpha > 0$, $\bm{\Lambda} \in \mathcal{M}_{m}(\mathbb{R})$, $\bm{B}$ an $m$-dimensional Brownian motion and  $\bm{X}$ a $m$-dimensional vector valued square-integrable stochastic process. Then, almost surely and for all $t \in [0,1]$
\begin{align*}
     & (-1)^{N-1} \bm{\Lambda}^{N}  I^{N\alpha}(\bm{X})_t \underset{N \to \infty}{\to} 0 \\
    &\sum_{n \geq N} (-1)^{n-1} \bm{\Lambda}^{n}  I^{n\alpha}_{\bm{B}}(\textnormal{diag}(\bm{X}))_t \underset{N \to \infty}{\to} 0.
\end{align*}
\end{lemma}

\begin{proof}
Let $N^* > N_{\alpha} := \lfloor \dfrac{1}{\alpha} \rfloor$. Since $\bm{X}$ is square-integrable, we have
\begin{align*}
    \mathbb{E}\Big[ \Big \lVert \sum_{N > N_{*}} \bm{\Lambda}^{N} I^{(N+1)\alpha}_{\bm{B}}(\textnormal{diag}(\bm{X}))_t \Big \rVert^2 \Big] & \leq  \sum_{N_1, N_2 > N_{*}} \mathbb{E}[\trans{(\bm{\Lambda}^{N_1} I^{(N_1+1)\alpha}_{\bm{B}}(\textnormal{diag}(\bm{X}))_t}) (\bm{\Lambda}^{N_2} I^{(N_2+1)\alpha}_{\bm{B}}(\textnormal{diag}(\bm{X}))_t)].
\end{align*}
Using the Cauchy-Schwartz inequality and writing $\normOp{\cdot}$ for the operator norm associated to the Euclidian norm, we find
\begin{align*}
    \mathbb{E}\Big[ \Big \lVert \sum_{N > N_{*}} \bm{\Lambda}^{N} I^{(N+1)\alpha}_{\bm{B}}(\textnormal{diag}(\bm{X}))_t \Big \rVert^2 \Big]  &\leq \sum_{N_1, N_2 > N_{*}}  \normOp{\bm{\Lambda}}^{N_1+N_2}  \sum_{1 \leq k,l \leq m} \mathbb{E}[ I^{(N_1+1)\alpha}_{B^k}(X^{k})_t I^{(N_2+1)\alpha}_{B^l}(X^{l})_t] \\
    & \leq \sum_{N_1, N_2 > N_{*}}  \normOp{\bm{\Lambda}}^{N_1+N_2}   \dfrac{1}{\Gamma((N_1+1)\alpha)\Gamma((N_2+1)\alpha)} \sum_{1 \leq i \leq m} \int_0^t (t-s)^{(N_1+N_2)\alpha - 2} \mathbb{E}[(X^i_s)^2] ds \\
    & \leq c \sum_{N_1, N_2 > N_{*}}  \dfrac{\normOp{\bm{\Lambda}}^{N_1+N_2}}{\Gamma((N_1+1)\alpha)\Gamma((N_2+1)\alpha)} \\
    & \leq c \Big( \sum_{N > N_{*}} \dfrac{\normOp{\bm{\Lambda}}^{N}}{\Gamma((N+1)\alpha)} \Big)^2.
\end{align*}
Thus by comparison of functions (for example by application of Stirling's formula), for all $\epsilon > 0$,
$$
\sum_{N > N_{\alpha}} \mathbb{P} \Big(  \Big \lVert \sum_{N > N_{*}} \bm{\Lambda}^{N} I^{(N+1)\alpha}_{\bm{B}}(\textnormal{diag}(\bm{X}))_t \Big \rVert > \epsilon \leq \dfrac{1}{\epsilon^2} \sum_{N_* \geq N_{\alpha}} \mathbb{E}\Big[ \Big \lVert \sum_{N > N_{*}} \bm{\Lambda}^{N} I^{(N+1)\alpha}_{\bm{B}}(\textnormal{diag}(\bm{X}))_t \Big \rVert^2 \Big]< \infty.
$$
The Borel-Cantelli lemma yields the almost sure convergence to zero of $\bm{\Lambda}^{N} I^{(N+1)\alpha}_{\bm{B}}(\textnormal{diag}(\bm{X}))$ as $N \to \infty$. The same approach yields the almost sure convergence to zero of $(-1)^{N-1} \bm{\Lambda}^{N}  I^{N\alpha}(\bm{X})$ as $N \to \infty$.
\end{proof}


\subsection{Proof of Corollary \ref{cor:correlation_vol}}
\label{sec:proof_example_corvol}

Take $\mu^1, \mu^2 > 0$, $\alpha \in (1/2,1), \kappa \in [0,1], H^b_{21}, H^s_{21}, H^b_{12},  H^s_{12} \in [0,1]$ such that (here $\sqrt{\cdot}$ is the principal square root, so that if $x<0$, $\sqrt{x}=i \sqrt{-x}$):
\begin{align*}
    &0 \leq (H_{12}^b + H_{12}^s)(H_{21}^b + H_{21}^s) < 1 \\
    &0 \leq \modulus{\kappa-\sqrt{(H^b_{12}-H^s_{12})(H^b_{21}-H^s_{21})}} < 1 \\
    &0 \leq \modulus{\kappa+\sqrt{(H^b_{12}-H^s_{12})(H^b_{21}-H^s_{21})}} < 1.
\end{align*}
Define now the following functions, for $t \geq 0$, which will appear in the structure of the kernel:
\begin{align*}
    \phi^T_1(t) &:= \alpha (1+\kappa/2) \mathbb{1}_{t \geq 1} t^{-(\alpha + 1)}& \phi_3^{b,T}(t) &= \alpha T^{-\alpha} H^b_{21} \mathbb{1}_{t \geq 1} t^{-(\alpha + 1)} \\
    \phi^T_2(t) &:=  \alpha (1+\kappa/2) \mathbb{1}_{t \geq 1} t^{-(\alpha + 1)}& \phi_3^{s,T}(t) &= \alpha T^{-\alpha} H^s_{21} \mathbb{1}_{t \geq 1} t^{-(\alpha + 1)}\\
    \lambda^T(t) &:= \alpha (\kappa-\kappa T^{-\alpha}) \mathbb{1}_{t \geq 1} t^{-(\alpha + 1)}& \phi_4^{b,T}(t) &:= \alpha T^{-\alpha} H^b_{12} \mathbb{1}_{t \geq 1} t^{-(\alpha + 1)}\\
    \tilde{\lambda}^T(t) &:= \alpha(\kappa-\kappa T^{-\alpha})\mathbb{1}_{t \geq 1} t^{-(\alpha + 1)}& \phi_4^{s,T}(t) &= \alpha T^{-\alpha} H^s_{12} \mathbb{1}_{t \geq 1} t^{-(\alpha + 1)}.
\end{align*}
The sequence of baselines and kernels are chosen as:
$$
\bm{\mu^T} = T^{\alpha - 1}
\begin{pmatrix}
\mu^{1} \\
\mu^{1} \\
\mu^{2} \\
\mu^{2}
\end{pmatrix}, \quad
\bm{\phi^T} = \begin{pmatrix}
    \phi_1^T & \phi_1^T - \lambda^T & \phi_3^{T,b} & \phi_3^{T,s} \\
    \phi_2^T - \lambda^T & \phi_2^T & \phi_3^{T,s} & \phi_3^{T,b} \\
    \phi_4^{T,b} & \phi_4^{T,s} & \phi_1^T & \phi_1^T - \tilde{\lambda}^T \\
    \phi_4^{T,s} & \phi_4^{T,b} & \phi_2^T - \tilde{\lambda}^T & \phi_2^T
\end{pmatrix}.
$$
The above sequence naturally satisfies the different assumptions outlined in Section \ref{sec:assumptions}. Indeed, using the following change of basis
$$
\bm{O} = \begin{pmatrix}
1 & 0 & 1 & 0\\
1 & 0 & -1 & 0\\
0 & 1 & 0 & 0\\
0 & 1 & 0 & -1
\end{pmatrix},
$$
we have, with notations from Section \ref{sec:assumptions},
\begin{align*}
    \bm{A} &= \begin{pmatrix}
        \phi_{1} + \phi_2 - \lambda & \phi_3^b + \phi_3^s \\
        \phi_4^b + \phi_4^s &\phi_1 + \phi_2 - \tilde{\lambda}
    \end{pmatrix} \\
    \bm{B} &= (\phi_1 - \phi_2)I \\
    \bm{C} &= \begin{pmatrix}
        \lambda & \phi_3^b - \phi_3^s \\
        \phi_4^b - \phi_4^s & \tilde{\lambda}
    \end{pmatrix} \\
    \bm{M} &= \alpha \bm{I} \\
    \bm{K} &= 
    \begin{pmatrix}
        \kappa & H^b_{21} +  H^s_{21} \\
        H^b_{12} +  H^s_{12} & \kappa
    \end{pmatrix}.
\end{align*}
Furthermore, we can check that the assumptions of Section \ref{sec:assumptions} are satisfied if
\begin{align*}
    &0 \leq H_{21}H_{12} < 1 \\
    &0 \leq \modulus{\kappa-\sqrt{(H^b_{12}-H^s_{12})(H^b_{21}-H^s_{21})}} < 1 \\
    &0 \leq \modulus{\kappa+\sqrt{(H^b_{12}-H^s_{12})(H^b_{21}-H^s_{21})}} < 1.
\end{align*}
Under those conditions, we can apply Theorem \ref{thm:ConvergenceInt} and compute the relevant quantities which appear in the limiting stochastic differential equation of volatility. We note for convenience
$$
\begin{pmatrix}
    x & y \\
    z & w
\end{pmatrix} := \Big( I-\norm{\bm{C}} \Big)^{-1} \norm{\bm{B}}.
$$
Then, straightforward linear algebra yields
\begin{align*}
    \bm{O_{11}} + \bm{O_{12}} \Big( I-\norm{\bm{C}} \Big)^{-1} \norm{\bm{B}} &= \begin{pmatrix}
        1 + x & y \\
        1 - x & -y
    \end{pmatrix} \\
    \Big( \bm{O_{11}} + \bm{O_{12}} \Big( I-\norm{\bm{C}} \Big)^{-1} \norm{\bm{B}} \Big)\bm{O^{(-1)}_{11}} &= \dfrac{1}{2} \begin{pmatrix}
        1+x & 1+x \\
        1-x & 1-x
    \end{pmatrix} \\
    \Big( \bm{O_{11}} + \bm{O_{12}} \Big(I-\norm{\bm{C}} \Big)^{-1} \norm{\bm{B}} \Big)\bm{O^{(-1)}_{12}} &= \dfrac{1}{2} \begin{pmatrix}
        y & y \\
        -y & -y
    \end{pmatrix},
\end{align*}
so that, using the notations of Theorem \ref{thm:ConvergenceInt} for the Brownian motion $\bm{B}$, $\bm{W^1}$ and $\bm{W^2}$, we have
\begin{align*}
    \Big( \bm{O_{11}} + \bm{O_{12}} \Big( I-\norm{\bm{C}} \Big)^{-1} \norm{\bm{B}} \Big)\bm{O^{(-1)}_{11}} \int_0^t \textnormal{diag}\Big( \sqrt{\bm{\Theta^1 \tilde{V}_s}} \Big) d\bm{W^1_s} &=  \dfrac{1}{2} \int_0^t \begin{pmatrix}
        (1+x) \big( \sqrt{V^1_t}dB^1_s + \sqrt{V^2_t}dB^2_s \big) \\
        (1-x) \big( \sqrt{V^1_t}dB^1_s + \sqrt{V^2_t}dB^2_s \big)
    \end{pmatrix} \\
    \Big( \bm{O_{11}} + \bm{O_{12}} \Big( I-\norm{\bm{C}} \Big)^{-1} \norm{\bm{B}} \Big)\bm{O^{(-1)}_{12}} \textnormal{diag}\Big (\sqrt{\bm{\Theta^2 \tilde{V}_s}} \Big) d\bm{W^2_t} &= \dfrac{1}{2} \int_0^t \begin{pmatrix}
        y \big( \sqrt{V^3_t}dB^3_t + \sqrt{V^4_t}dB^4_t \big) \\
        -y \big( \sqrt{V^3_t}dB^3_t + \sqrt{V^4_t}dB^4_t \big)
    \end{pmatrix}.
\end{align*}

Therefore, writing 
$$
\bm{\Sigma}_{1} := \dfrac{1}{2}\begin{pmatrix}
    1 + x & 1+x & y & y  \\
    1 - x & 1 - x & -y & - y
\end{pmatrix},
$$
we have the following equation for the fundamental variance of Asset $1$
\begin{align*}
    \dfrac{\Gamma(1-\alpha) \Gamma(\alpha)}{\alpha} \begin{pmatrix}
        V^1_t \\
        V^2_t
    \end{pmatrix} &= \int_0^t (t-s)^{\alpha - 1} \Big[ \begin{pmatrix}
        1 + x & y \\
        1 - x & -y
    \end{pmatrix} \begin{pmatrix}
        \mu_1 \\
        \mu_2
    \end{pmatrix} -  \begin{pmatrix}
        1 + x & y \\
        1 - x & -y
    \end{pmatrix} \bm{K^{-1}} \begin{pmatrix}
        1 + x & y \\
        1 - x & -y
    \end{pmatrix}^{-1} \begin{pmatrix}
        V^1_s \\
        V^2_s 
    \end{pmatrix} \Big] ds \\
    &+ \int_0^t (t-s)^{\alpha - 1}
    \bm{\Sigma}_{1}
    \textnormal{diag}(\sqrt{\bm{V}_s}) dB_s.
\end{align*}
By symmetry, we can find the analogue to the above on the second asset. Using the following notations
\begin{align*}
    &\bm{\Sigma} := \dfrac{\alpha}{\Gamma(1-\alpha) \Gamma(\alpha)}
\dfrac{1}{2}
\begin{pmatrix}
    1 + x & 1+x & y & y  \\
    1 - x & 1 - x & -y & - y \\
    z & z & 1+w & 1+w  \\
     -z & -z & 1-w & 1-w
\end{pmatrix}, \quad \bm{D}:= \dfrac{\alpha}{\Gamma(1-\alpha) \Gamma(\alpha)} \begin{pmatrix}
        1 + x & y \\
        1 - x & -y \\
        1 + w & z \\
        1 - w & -z
    \end{pmatrix},\\
    &\bm{G}:= \dfrac{\alpha}{\Gamma(1-\alpha) \Gamma(\alpha)} \begin{pmatrix}
        \begin{pmatrix}
        1 + x & y \\
        1 - x & -y
    \end{pmatrix} K^{-1} \begin{pmatrix}
        1 + x & y \\
        1 - x & -y
    \end{pmatrix}^{-1} & 0 \\
    0 & \begin{pmatrix}
        z & 1+w \\
        -z & 1-w
    \end{pmatrix} K^{-1} \begin{pmatrix}
        z & 1+w \\
        -z & 1-w
    \end{pmatrix}^{-1}
    \end{pmatrix},
\end{align*}
where we have written for convenience
\begin{align*}
    \begin{pmatrix}
    x & y \\
    z & w
\end{pmatrix} &:= \dfrac{\norm{\phi_1} - \norm{\phi_2}}{\norm{\lambda}\norm{\tilde{\lambda}}- (\norm{\phi^b_4} - \norm{\phi^s_4}) (\norm{\phi^b_3} - \norm{\phi^s_3})} \begin{pmatrix}
    \norm{\tilde{\lambda}} & -(\norm{\phi^b_3} - \norm{\phi^s_3}) \\
    -(\norm{\phi^b_4} - \norm{\phi^s_4}) & \norm{\lambda}
\end{pmatrix}.
\end{align*}
Therefore $\bm{V}$ satisfies the following stochastic Volterra equation
\begin{align*}
    \dfrac{\Gamma(1-\alpha) \Gamma(\alpha)}{\alpha}\bm{V_t} &= \int_0^t (t-s)^{\alpha - 1} \Big[ \bm{D}
    \begin{pmatrix}
        \mu_1 \\
        \mu_2
    \end{pmatrix} -  \bm{G} \bm{V_s} \Big] ds + \int_0^t (t-s)^{\alpha - 1}
    \bm{\Sigma}
    \textnormal{diag}(\sqrt{\bm{V_s}}) d\bm{B_s}.
\end{align*}
This concludes the proof of Corollary \ref{cor:correlation_vol}.

\subsection{Proof of Corollary \ref{cor:example_2D}}
\label{sec:proof_example_2d}

We split the proof into two steps. First, we show that the structure of the kernel satisfies the assumptions of Section \ref{sec:assumptions}. Then we compute the equations satisfied by variance and prices.

\subsubsection*{Checking for the assumptions of Theorem \ref{thm:ConvergenceInt}}
We write
$$
\bm{O}_1 := \begin{pmatrix}
1 \\
1 \\
0 \\
0
\end{pmatrix}
\bm{O}_2 := \begin{pmatrix}
0 \\
0 \\
1 \\
1
\end{pmatrix}
\bm{O}_3 := \begin{pmatrix}
1 \\
-1 \\
0 \\
0
\end{pmatrix}
\bm{O}_4 := \begin{pmatrix}
0 \\
0 \\
1 \\
-1
\end{pmatrix}.
$$
Then, setting $\bm{O} := \begin{pmatrix}
\bm{O}_1 \mid \bm{O}_2 \mid \bm{O}_3 \mid \bm{O}_4
\end{pmatrix}$, we have
\begin{align*}
    \bm{\phi^T} &= \bm{O} \begin{pmatrix}
    \phi_1^T + \phi_2^T & \phi^{T,c}_{12} + \phi^{T,a}_{12} & 0 & 0 \\
    \phi^b_{21} + \phi^s_{21} & \tilde{\phi}^T_1 +  \tilde{\phi}^T_2 & 0 & 0 \\
    0 & 0 & \phi_1^T - \phi_2^T & \phi^{T,c}_{12} - \phi^{T,a}_{12} \\
    0 & 0 & \phi^b_{21} - \phi^s_{21} & \tilde{\phi}^T_1 - \tilde{\phi}^T_2
    \end{pmatrix} \bm{O}^{-1}.
\end{align*}
It is straightforward to check that the assumptions are satisfied if
\begin{align*}
    0 &\leq (H^c_{12}+H^a_{12})(H^c_{21}+H^a_{21}) < 1 \\
    0 &\leq \modulus{1 - (\gamma_1 + \gamma_2) - \sqrt{(H^c_{12}-H^a_{12})(H^c_{21}-H^a_{21}) + (\gamma_1-\gamma_2)^2}} < 1 \\
    0 &\leq \modulus{1 - (\gamma_1 + \gamma_2) + \sqrt{(H^c_{12}-H^a_{12})(H^c_{21}-H^a_{21}) + (\gamma_1-\gamma_2)^2}} < 1.
\end{align*}

Under those conditions $\bm{K}=\bm{I}-\bm{H}$ has positive eigenvalues and therefore $\bm{KM}^{-1} = \dfrac{1}{\alpha}\bm{K}$  has positive eigenvalues. Therefore all the assumptions of Theorem \ref{thm:ConvergenceInt} are satisfied.

\subsubsection*{Limiting variance process}

Since we can apply Theorem \ref{thm:ConvergenceInt}, we now compute the relevant quantities. As $\bm{B} = \bm{0}$, writing $H^{12} := H^{a}_{12} + H^{c}_{12}$ and  $H^{21} := H^{a}_{21} + H^{c}_{21}$, we have
\begin{align*}
    &\bm{O}^{-1} = \dfrac{1}{2}\begin{pmatrix}
    1 & 1 & 0 & 0 \\
    0 & 0 & 1 & 1 \\
    1 & -1 & 0 & 0 \\
    0 & 0 & 1 & -1
    \end{pmatrix} \quad  \bm{K}^{-1} = \dfrac{1}{1 - H_{12}H_{21}}
    \begin{pmatrix}
    1 & H_{12} \\
    H_{21} & 1
    \end{pmatrix} \\
    &\bm{\Theta^1} = \dfrac{1}{1 - H_{12}H_{21}} \begin{pmatrix}
    1 & H_{12} \\
    1 & H_{12}
    \end{pmatrix} \quad \bm{\Theta^2} = \dfrac{1}{1 - H_{12}H_{21}} \begin{pmatrix}
    H_{21} & 1 \\
    H_{21} & 1
    \end{pmatrix}.
\end{align*}
One can check that the equations satisfied by $\bm{\Theta^1} \bm{\tilde{V}}$ and $\bm{\Theta^2} \bm{\tilde{V}}$ are, where $\bm{B}$ is a Brownian motion,
\begin{align*}
    \bm{\Theta^1} \bm{\tilde{V}_t} &= \dfrac{\alpha}{\Gamma(\alpha)\Gamma(1 - \alpha)} \int_0^t (t-s)^{\alpha-1} \left[
    \begin{pmatrix}
    \mu_1 \\
    \mu_1
    \end{pmatrix} - \begin{pmatrix}
    \tilde{V}^1_s \\
    \tilde{V}^1_s
    \end{pmatrix} \right]ds \\
    &+ \dfrac{\alpha}{\Gamma(\alpha)\Gamma(1 - \alpha)} \int_0^t (t-s)^{\alpha-1} \sqrt{\tilde{V}^1_s + H_{12}\tilde{V}^2_s} \begin{pmatrix}
    dB^{1}_s+dB^{2}_s \\
     dB^{1}_s+dB^{2}_s
    \end{pmatrix} \\
    \bm{\Theta^2} \bm{\tilde{V}_t} &=  \dfrac{\alpha}{\Gamma(\alpha)\Gamma(1 - \alpha)} \int_0^t (t-s)^{\alpha-1} \left[
    \begin{pmatrix}
    \mu_2 \\
    \mu_2
    \end{pmatrix} - \begin{pmatrix}
    \tilde{V}^2_s \\
    \tilde{V}^2_s
    \end{pmatrix} \right]ds \\
    &+ \dfrac{\alpha}{\Gamma(\alpha)\Gamma(1 - \alpha)} \int_0^t (t-s)^{\alpha-1} \sqrt{\tilde{V}^2_s + H_{21}\tilde{V}^1_s} \begin{pmatrix}
    dB^{3}_s+dB^{4}_s \\
    dB^{3}_s+dB^{4}_s
    \end{pmatrix}.
\end{align*}

Note that the above implies that $V^{1+} = V^{1-}$ and $V^{2+} = V^{2-}$. This property is due to the the symmetric structure of the baselines and kernels. Therefore, the joint dynamics can be fully captured by considering the joint dynamics of $(V^{1+}, V^{2+})$. Thus, writing $V^{1} := V^{1+}=V^{1-}$ and $V^{2} := V^{2+} = V^{2-}$, we have
\begin{align*}
    \Gamma(\alpha) \dfrac{\Gamma(1 - \alpha)}{\alpha} V^{1}_t &= \int_0^t (t-s)^{\alpha - 1} (\mu_1 - \tilde{V}^{1}_s)ds +  \int_0^t \sqrt{V^{1}_t} (dB^{1}_s+dB^{2}_s) \\
    \Gamma(\alpha) \dfrac{\Gamma(1 - \alpha)}{\alpha}V^{2}_t &= \int_0^t (t-s)^{\alpha - 1} (\mu_2 - \tilde{V}^{2}_s)ds +  \int_0^t \sqrt{V^{2}_t} (dB^{3}_s+dB^{4}_s).
\end{align*}
We can write the above without $\bm{\tilde{V}}$ as
\begin{align*}
    \Gamma(\alpha) \dfrac{\Gamma(1 - \alpha)}{\alpha} \begin{pmatrix} V^{1}_t \\ V^{2}_t \end{pmatrix} = \int_0^t (t-s)^{\alpha - 1} \big( \begin{pmatrix}
    \mu_1 \\
    \mu_2 
    \end{pmatrix} - \bm{K}^{-1} \begin{pmatrix} V^{1}_s \\ V^{2}_s \end{pmatrix} \big)ds + \int_0^t (t-s)^{\alpha - 1} \begin{pmatrix}
    \sqrt{V^1_s} (dB^{11}_s + dB^{12}_s) \\
    \sqrt{V^2_s} (dB^{21}_s + dB^{22}_s)
    \end{pmatrix}.
\end{align*}

\subsubsection*{Limiting price process}
Turning now to the price process, it remains to compute $\bm{\Delta}$ (see Equation \eqref{def:Delta_matrix}) using the definition. We have
\begin{align*}
    \trans{\normOne{\bm{\psi}}}\bm{O_3} &= \sum_{k \geq 1} \trans{\normOne{\bm{\phi}}}^k
    \bm{O_3} \\
    &= \bm{O} \sum_{k \geq 1} \big[ \big( \norm{\bm{C}} \big)^{k}_{11} \bm{e_3} + \big( \norm{\bm{C}} \big)^{k}_{12} \bm{e_4} \big] \\
    &= \sum_{k \geq 1} \big[ \big( \norm{\bm{C}} \big)^{k}_{11} \bm{O_3} + \big( \norm{\bm{C}} \big)^{k}_{12} \bm{O_4} \big]  \\
    &= [(\bm{I} - \norm{\bm{C}})^{-1} - \bm{I}]_{11} \bm{O_3} + [(\bm{I} - \norm{\bm{C}})^{-1} - \bm{I}]_{12} \bm{O_4},
\end{align*}
which, by definition of $\bm{\Delta}$, yields
\begin{align*}
       \Delta_{11} &= \big[ \big(\bm{I} - \norm{\bm{C}} \big)^{-1} - \bm{I} \big]_{11} = \dfrac{2 \gamma_2}{4 \gamma_1 \gamma_2 - (H_{12}^c - H_{12}^a)(H_{21}^c - H_{21}^a)} - 1 \\
    \Delta_{12} &= \big[ \big(\bm{I} - \norm{\bm{C}} \big)^{-1} - \bm{I} \big]_{12}  =\dfrac{H_{21}^c - H_{21}^a}{4 \gamma_1 \gamma_2 - (H_{12}^c - H_{12}^a)(H_{21}^c - H_{21}^a)}.
\end{align*}
Therefore,
\begin{align*}
    \bm{\Delta} = \dfrac{1}{4 \gamma_1 \gamma_2 - (H_{12}^c - H_{12}^a)(H_{21}^c - H_{21}^a)}\begin{pmatrix}
   2 \gamma_2 & H_{21}^c - H_{21}^a \\
   H_{12}^c - H_{12}^a & 2 \gamma_1
    \end{pmatrix} - \bm{I}.
\end{align*}
Finally, any limit point $\bm{P}$ of the sequence of microscopic price processes satisfies the following equation
\begin{align*}
    \bm{P}_t &= \dfrac{1}{4 \gamma_1 \gamma_2 - (H_{12}^c - H_{12}^a)(H_{21}^c - H_{21}^a)}\begin{pmatrix}
   2 \gamma_2 & H_{21}^c - H_{21}^a \\
   H_{12}^c - H_{12}^a & 2 \gamma_1
    \end{pmatrix}  
    \begin{pmatrix}
    1 & -1 & 0 & 0 \\
    0 & 0 & 1 & -1
    \end{pmatrix}
    \int_0^t \begin{pmatrix}
    \sqrt{V^1_s}dB^{1}_s \\
    \sqrt{V^1_s}dB^{2}_s \\
    \sqrt{V^2_s}dB^{3}_s \\
    \sqrt{V^2_s}dB^{4}_s
    \end{pmatrix} \\
    &= \dfrac{1}{4 \gamma_1 \gamma_2 - (H_{12}^c - H_{12}^a)(H_{21}^c - H_{21}^a)}\begin{pmatrix}
   2 \gamma_2 & H_{21}^c - H_{21}^a \\
   H_{12}^c - H_{12}^a & 2 \gamma_1
    \end{pmatrix}
    \int_0^t
    \begin{pmatrix}
    \sqrt{V^1_s} (dB^{1}_s -  dB^{2}_s) \\
    \sqrt{V^2_s} (dB^{3}_s -  dB^{4}_s)
    \end{pmatrix}.
\end{align*}
This concludes the proof of Corollary \ref{cor:example_2D}.

\subsection{Proof of Corollary \ref{cor:example_nD}}
\label{sec:proof_example_nd}

We define the interaction kernel between Asset $i$ and Asset $j$, for $1 \leq i,j \leq m$, define
$$
\bm{\phi^T_{ij}(t)} := \left\{
    \begin{array}{ll}
        \alpha (1-T^{-\alpha}) \mathbb{1}_{t \geq 1} t^{-(\alpha + 1)}  \begin{pmatrix}
        (1-\gamma) & \gamma \\
        \gamma& (1-\gamma)
        \end{pmatrix} & \mbox{if } i=j,\\
        \alpha T^{-\alpha}  \mathbb{1}_{t \geq 1} t^{-(\alpha + 1)}  
        \begin{pmatrix}
        H^{c} & H^{a} \\
        H^{a} & H^{c}
        \end{pmatrix} & \mbox{if Asset } i \mbox{ and Asset } j \mbox{ belong to the same sector}, \\
         \alpha T^{-\alpha}  \mathbb{1}_{t \geq 1} t^{-(\alpha + 1)}  
        \begin{pmatrix}
        H^{c} + H_r^{c} & H^{a} + H_r^{a} \\
        H^{a} + H_r^{a} & H^{c} + H_r^{c}
        \end{pmatrix} & \mbox{otherwise.}
    \end{array}
\right.
$$
Finally, the complete Hawkes baseline and kernel structure is
$$
\bm{\mu^T} = T^{\alpha - 1}
\begin{pmatrix}
\mu^{1} \\
\mu^{1} \\
\vdots \\
\mu^{m} \\
\mu^{m} \\
\end{pmatrix}, \quad
\bm{\phi^T} = \begin{pmatrix}
\bm{\phi_{11}^T} & \bm{\phi_{12}^T} & \hdots &  \bm{\phi_{1m}^T} \\
\bm{\phi_{21}^T} & \bm{\phi_{22}^T}  & \hdots & \bm{\phi_{2m}^T} \\
\vdots & \hdots & \ddots & \vdots \\
\bm{\phi_{m1}^T}  & \hdots &  \hdots  &   \bm{\phi_{mm}^T} \\
\end{pmatrix}.
$$
As in the previous example, the proof is split into three steps. First, we show that the structure of the kernel satisfies the assumptions required to apply Theorem \ref{thm:ConvergenceInt}. Then, we compute the equation satisfied by the variance and finally the limiting price process.

\subsubsection*{Checking assumptions of Theorem \ref{thm:ConvergenceInt}}

We can examine the structure of the kernel as in the two-asset example. Define the following basis:
\begin{align*}
    \bm{O_{i}} := \left\{
    \begin{array}{ll}
        \bm{e_{2i}}+\bm{e_{2i+1}} & \mbox{if } 1 \leq i \leq m, \\
        \bm{e_{2i}}-\bm{e_{2i}} & \mbox{if } m+1 \leq i \leq 2m.
    \end{array}
\right.
\end{align*}
Using the notations of Section \ref{sec:assumptions}, straightforward computations allow us to write
\begin{align*}
    \bm{\phi^T} &= \bm{O} \begin{pmatrix}
    \bm{A^T} & \bm{0} \\
    \bm{B^T} & \bm{C^T}
    \end{pmatrix} \bm{O}^{-1} = \bm{O} \begin{pmatrix}
    \bm{A^T} & \bm{0} \\
    \bm{0} & \bm{C^T}
    \end{pmatrix} \bm{O}^{-1},
\end{align*}
where we can compute $\bm{A^T}$ and $\bm{C^T}$. Checking the assumptions is done as in the two-asset case, though the conditions have changed here due to the new structure of the kernel. For example, since
\begin{equation*}
    \underset{T \to \infty}{\lim} \norm{\bm{\phi^T}} \bm{O}_{m+i} =  (1-2\gamma) \bm{O}_{n+i} + (H^c - H^a) \sum_{1 \leq j \neq i \leq m} \bm{O}_{m+j} + \sum_{1 \leq j \neq i \leq m} \sum_{1 \leq r \leq R} (H^{c}_r - H^{a}_r) \bm{O}_{m+j},
\end{equation*}
we have, writing $\bm{J} := \bm{e_1} \trans{\bm{e_1}} + \cdots + \bm{e_m} \trans{\bm{e_m}}$ and for any $1 \leq r \leq R$, $\bm{J_r} := \bm{e_{i_r}}\trans{\bm{e_{i_r}}} + \cdots +  \bm{e_{i_r + m_r}} \trans{\bm{e_{i_r + m_r}}}$,
$$\norm{\bm{C}} = (1-2\gamma) \bm{I} + (H^c - H^a) \bm{J} + \sum_{1 \leq r \leq R} (H^{c}_{r} - H^{a}_{r}) \bm{J_r}.$$
Therefore, as the eigenvalues of $\norm{\bm{C}}$ can be made explicit, if
\begin{align*}
    \modulus{\lambda^{-} + \sum_{1 \leq r \leq R} \lambda^{-}_r} < 2 \gamma,
\end{align*}
then $\spectralRadius{\norm{\bm{C^T}}} < 1$ and $\spectralRadius{\norm{\bm{C}}} < 1$. Similarly, we can easily check that a necessary condition for $\spectralRadius{\norm{\bm{A^T}}} < 1$ for $T$ large enough is
\begin{align*}
    \modulus{H^c + H^a + \sum_{1 \leq r \leq R} \dfrac{m_{r}-1}{m-1}(H^{c}_{r} + H^{a}_{r})} < \dfrac{1}{m-1}.
\end{align*}
Since we are interested in the limit where the number of assets grows to infinity, we also impose
\begin{align*}
    \modulus{\lambda^{-} + \sum_{1 \leq r \leq R} \eta_r \lambda^{-}_r} &< 2 \gamma \\
     \modulus{\lambda^{+} + \sum_{1 \leq r \leq R} \eta_r \lambda^{+}_r} &< 1.
\end{align*}
Combined, we have verified all the assumptions on the structure of the kernel. We thus move to assumptions on $\bm{K}$ and $\bm{\Lambda} := \bm{K} \bm{M}^{-1}$. As in the two-asset example, we have here $\bm{M} = \alpha \bm{I}$. Since $\bm{K} = \bm{I} - (H^{c}+H^{a})\bm{J} - \sum_{1 \leq r \leq R} (H^{c}_{r}+H^{a}_{r})\bm{J_r}$, the eigenvalues of $\bm{K}$ (and therefore those of $\bm{\Lambda}$) are all strictly positive. Thus we have checked all necessary conditions to apply Theorem \ref{thm:ConvergenceInt}. We can thus state the equation satisfied by the variance process.

\subsubsection*{Limiting variance process}

As in the previous example, we have $V^{i+} = V^{i-}$. Thus, we write the underlying variance of asset $i$ $V^i$ and use the (slight) abuse of notation and define $\bm{V} := (V^1, V^2, \cdots, V^{m})$. Then $\bm{V}$ satisfies
\begin{align*}
    \bm{V_t} &= \dfrac{\alpha}{\Gamma(\alpha)\Gamma(1 - \alpha)} \int_0^t (t-s)^{\alpha-1} \left[
    \bm{\theta} - \bm{K}^{-1}  \bm{V_s} \right]ds + \dfrac{\alpha \sqrt{2}}{\Gamma(\alpha)\Gamma(1 - \alpha)} \int_0^t (t-s)^{\alpha-1} \textnormal{diag}(\sqrt{\bm{V_s}}) d\bm{B_s},
\end{align*}
where $\bm{B}$ is a Brownian motion. We can rewrite $\bm{K}^{-1}$ as
\begin{align*}
    \bm{K}^{-1} &= \left( \bm{I} - (H^c + H^a) \bm{J} - \sum_{1 \leq r \leq R} (H^c_r + H^a_r) \bm{J_r} \right)^{-1},\\
    &= \left( \bm{I} - (H^c + H^a) (m-1) \bm{w} \trans{\bm{w}} - \sum_{1 \leq r \leq R} (H^c_r + H^a_r) (m_r - 1) \bm{w_r} \trans{\bm{w_r}} - \epsilon \right) ^{-1},
\end{align*}
with the small term $\bm{\epsilon}$
$$\bm{\epsilon} := (H^c + H^a)(\bm{J}-(m-1)\bm{w} \trans{\bm{w}}) + \sum_{1 \leq r \leq R} (H^c_r + H^a_r) (\bm{J_r}-(m_r - 1)\bm{w_r} \trans{\bm{w_r}}).
$$
It is easy to check that $\spectralRadius{\bm{\epsilon}} \underset{m \to \infty}{=} o(\dfrac{1}{m})$, which concludes our study of the variance process. We now turn to the equation satisfied by the limiting price process.

\subsubsection*{Limiting price process}

Using the same approach as in the two-asset case, computing $\bm{\Delta}$ boils down to computing $(I-\norm{\bm{C}})^{-1}$. Using the expression for $\norm{\bm{C}}$ derived previously, we have
$$
(\bm{I} - \bm{C})^{-1} = \dfrac{1}{2\gamma} (\bm{I} - \dfrac{H^c - H^a}{2\gamma} \bm{J} - \sum_{1 \leq r \leq R} \dfrac{H^c_r - H^a_r}{2\gamma} \bm{J_r})^{-1}.
$$
Therefore, repeating the same approach we used for $\bm{K}^{-1}$ yields
$$
(\bm{I} - \bm{C})^{-1} = (2 \gamma \bm{I} - \lambda^{-} \bm{w} \trans{\bm{w}} - \sum_{1 \leq r \leq R} \eta_r \lambda^{-}_r \bm{w_r} \trans{\bm{w_r}}-\bm{\epsilon})^{-1},
$$
with $\spectralRadius{\bm{\epsilon}} = o(\dfrac{1}{m})$. Thus, we have the expression of $\bm{\Delta}$
\begin{align*}
    \bm{\Delta} & = (2 \gamma \bm{I} - \lambda^{-} \bm{w} \trans{\bm{w}} - \sum_{1 \leq r \leq R} \eta_r \lambda^{-}_r \bm{w_r} \trans{\bm{w_r}} - \bm{\epsilon})^{-1} - \bm{I}.
\end{align*}
Plugging this into Theorem \ref{thm:ConvergenceInt}, we have the equation satisfied by macroscopic prices, which concludes the proof of Corollary \ref{cor:example_nD}.

\bibliographystyle{plain}
\bibliography{bibliography.bib}

\begin{thebibliography}{10}

\bibitem{Bacry2013ModellingProcesses}
Emmanuel Bacry, Sylvain Delattre, Marc Hoffmann, and Jean-François Muzy.
\newblock {Modelling microstructure noise with mutually exciting point
  processes}.
\newblock {\em Quantitative finance}, 13(1):65--77, 2013.

\bibitem{Bacry2015HawkesFinance}
Emmanuel Bacry, Iacopo Mastromatteo, and Jean-François Muzy.
\newblock {Hawkes processes in finance}.
\newblock {\em Market Microstructure and Liquidity}, 1(01):1550005, 2015.

\bibitem{Bayer2015PricingVolatility}
Christian Bayer, Peter~K. Friz, and Jim Gatheral.
\newblock {Pricing Under Rough Volatility}.
\newblock {\em SSRN Electronic Journal}, 9 2015.

\bibitem{Benzaquen2017DissectingAnalysis}
Michael Benzaquen, Iacopo Mastromatteo, Zoltan Eisler, and Jean-Philippe
  Bouchaud.
\newblock {Dissecting cross-impact on stock markets: An empirical analysis}.
\newblock {\em Journal of Statistical Mechanics: Theory and Experiment},
  2017(2):23406, 2017.

\bibitem{Cuchiero2019MarkovianProcesses}
Christa Cuchiero and Josef Teichmann.
\newblock {Markovian lifts of positive semidefinite affine Volterra type
  processes}.
\newblock {\em arXiv preprint arXiv:1907.01917}, 2019.

\bibitem{DaFonseca2019VolatilityRough}
José Da~Fonseca and Wenjun Zhang.
\newblock {Volatility of volatility is (also) rough}.
\newblock {\em Journal of Futures Markets}, 39(5):600--611, 2019.

\bibitem{Dandapani2019FromEffect}
Aditi Dandapani, Paul Jusselin, and Mathieu Rosenbaum.
\newblock {From quadratic Hawkes processes to super-Heston rough volatility
  models with Zumbach effect}.
\newblock {\em arXiv preprint arXiv:1907.06151}, 2019.

\bibitem{ElEuch2018TheVolatility.}
Omar El~Euch, {Fukasawa Masaaki}, and Mathieu Rosenbaum.
\newblock {The microstructural foundations of leverage effect and rough
  volatility.}
\newblock {\em Finance and Stochastics}, 22(2):241--280, 2018.

\bibitem{ElEuch2018TheHeston}
Omar El~Euch, Jim Gatheral, Rados Radoicic, and Mathieu Rosenbaum.
\newblock {The Zumbach effect under rough Heston}.
\newblock {\em To appear in Quantitative Finance}, 2018.

\bibitem{ElEuch2018RougheningHeston}
Omar El~Euch, Jim Gatheral, and Mathieu Rosenbaum.
\newblock {Roughening heston}.
\newblock {\em Available at SSRN 3116887}, 2018.

\bibitem{GJR14}
Jim Gatheral, Thibault Jaisson, and Mathieu Rosenbaum.
\newblock {Volatility is rough}.
\newblock {\em Quantitative Finance}, 18(6):933--949, 2018.

\bibitem{Hardiman2013CriticalAnalysis}
Stephen~J Hardiman, Nicolas Bercot, and Jean-Philippe Bouchaud.
\newblock {Critical reflexivity in financial markets: a Hawkes process
  analysis}.
\newblock {\em The European Physical Journal B}, 86(10):442, 2013.

\bibitem{Hawkes1971PointProcesses.}
Alan.~G Hawkes.
\newblock {Point spectra of some mutually exciting point processes.}
\newblock {\em Journal of the Royal Statistical Society: Series B
  (Methodological)}, 33(3):438--443, 1971.

\bibitem{Hawkes1971SpectraProcesses}
Alan~G Hawkes.
\newblock {Spectra of some self-exciting and mutually exciting point
  processes}.
\newblock {\em Biometrika}, 58(1):83--90, 1971.

\bibitem{Hawkes1974AProcess}
Alan~G Hawkes and David Oakes.
\newblock {A cluster process representation of a self-exciting process}.
\newblock {\em Journal of Applied Probability}, 11(3):493--503, 1974.

\bibitem{Higham2008FunctionsComputation}
Nicholas~J Higham.
\newblock {\em {Functions of matrices: theory and computation}}, volume 104.
\newblock Siam, 2008.

\bibitem{Horvath2019DeepVolatility}
Blanka Horvath, Aitor Muguruza, and Mehdi Tomas.
\newblock {Deep learning volatility}.
\newblock {\em Available at SSRN 3322085}, 2019.

\bibitem{Jaber2019AType}
Eduardo~Abi Jaber, Christa Cuchiero, Martin Larsson, and Sergio Pulido.
\newblock {A weak solution theory for stochastic Volterra equations of
  convolution type}.
\newblock {\em arXiv preprint arXiv:1909.01166}, 2019.

\bibitem{Jaber2017AffineProcesses}
Eduardo~Abi Jaber, Martin Larsson, and Sergio Pulido.
\newblock {Affine volterra processes}.
\newblock {\em To appear in Annals of Applied Probability}, 2017.

\bibitem{Jaisson2015LimitProcesses}
Thibault Jaisson and Mathieu Rosenbaum.
\newblock {Limit theorems for nearly unstable Hawkes processes}.
\newblock {\em The Annals of Applied Probability}, 25(2):600--631, 2015.

\bibitem{Jaisson2016RoughProcesses}
Thibault Jaisson and Mathieu Rosenbaum.
\newblock {Rough fractional diffusions as scaling limits of nearly unstable
  heavy tailed Hawkes processes}.
\newblock {\em The Annals of Applied Probability}, 26(5):2860--2882, 2016.

\bibitem{Jusselin2018No-ArbitrageVolatility}
Paul Jusselin and Mathieu Rosenbaum.
\newblock {No-arbitrage implies power-law market impact and rough volatility}.
\newblock {\em arXiv preprint arXiv:1805.07134}, 2018.

\bibitem{Laloux1999NoiseMatrices}
Laurent Laloux, Pierre Cizeau, Jean-Philippe Bouchaud, and Marc Potters.
\newblock {Noise dressing of financial correlation matrices}.
\newblock {\em Physical review letters}, 83(7):1467, 1999.

\bibitem{Livieri2018RoughPrices}
Giulia Livieri, Saad Mouti, Andrea Pallavicini, and Mathieu Rosenbaum.
\newblock {Rough volatility: evidence from option prices}.
\newblock {\em IISE Transactions}, 50(9):767--776, 2018.

\bibitem{Reigneron2011PrincipalEffect}
Pierre-Alain Reigneron, Romain Allez, and Jean-Philippe Bouchaud.
\newblock {Principal regression analysis and the index leverage effect}.
\newblock {\em Physica A: Statistical Mechanics and its Applications},
  390(17):3026--3035, 2011.

\bibitem{Revuz2013ContinuousMotion}
Daniel Revuz and Marc Yor.
\newblock {\em {Continuous martingales and Brownian motion}}, volume 293.
\newblock Springer Science {\&} Business Media, 2013.

\bibitem{Veraar2012TheRevisited}
Mark Veraar.
\newblock {The stochastic Fubini theorem revisited}.
\newblock {\em Stochastics An International Journal of Probability and
  Stochastic Processes}, 84(4):543--551, 2012.

\end{thebibliography}

\end{document}